\documentclass[a4paper,UKenglish,cleveref, autoref, thm-restate,nolineno]{socg-lipics-v2021}

\usepackage{amsthm,amsmath,amssymb,amsfonts}
\usepackage{todonotes}
\usepackage{framed}
\usepackage{cleveref}
\usepackage{boxedminipage}
 \usepackage{tikz}
 \usepackage{pgfmath}
\usepackage{dsfont}
\usepackage{xspace}
\usepackage{graphicx}
\usepackage[numbers]{natbib}
\usepackage{nicefrac}
\usepackage{todonotes} 
\setuptodonotes{inline}
\usepackage{xcolor} 
\usepackage{tikz}
\usetikzlibrary{calc}
\usepackage{boxedminipage}
\usepackage{framed}
\usepackage{thm-restate}
\usepackage{tabularx}
\usepackage{tcolorbox}
\usepackage{etoolbox}
\usepackage{xifthen}
\usepackage{listings}

\newtheorem{reduction}{Reduction Rule}

\newcommand{\shortversion}[1]{}
\newcommand{\deferProof}[0]{\shortversion{\rm{ ($\star$)}}}
\newcommand{\longversion}[1]{#1}

\newcommand{\mcm}[3]{\newcommand{#1}[#2]{{\ensuremath{#3}}}} 
\mcm{\Nbb}{0}{\mathbb{N}}
\mcm{\Zbb}{0}{\mathbb{Z}}
\mcm{\Rbb}{0}{\mathbb{R}}
\mcm{\Cbb}{0}{\mathbb{C}}
\mcm{\Qbb}{0}{\mathbb{Q}}
\mcm{\Acal}{0}{\cal A}
\mcm{\Bcal}{0}{\cal B}
\mcm{\Ccal}{0}{\cal C}
\mcm{\Dcal}{0}{\cal D}
\mcm{\Ecal}{0}{\cal E}
\mcm{\Fcal}{0}{\cal F}
\mcm{\Gcal}{0}{\cal G}
\mcm{\Hcal}{0}{\cal H}
\mcm{\Ical}{0}{\cal I}
\mcm{\Jcal}{0}{\cal J}
\mcm{\Kcal}{0}{\cal K}
\mcm{\Lcal}{0}{\cal L}
\mcm{\Mcal}{0}{\cal M}
\mcm{\Ncal}{0}{\cal N}
\mcm{\Ocal}{0}{{\cal O}}
\mcm{\Pcal}{0}{{\cal P}}
\mcm{\Qcal}{0}{{\cal Q}}
\mcm{\Rcal}{0}{{\cal R}}
\mcm{\Scal}{0}{{\cal S}}
\mcm{\Tcal}{0}{{\cal T}}
\mcm{\Ucal}{0}{{\cal U}}
\mcm{\Vcal}{0}{{\cal V}}
\mcm{\Wcal}{0}{{\cal W}}
\mcm{\Xcal}{0}{{\cal X}}
\mcm{\Ycal}{0}{{\cal Y}}
\mcm{\Zcal}{0}{{\cal Z}}

\newcommand{\Oh}{\mathcal{O}}
\newcommand{\koutl}{k_{O}}
\newcommand{\kmod}{k_{M}}
\newcommand{\woutl}{w_{O}}
\newcommand{\wmod}{w_{M}}
\newcommand{\bigoh}[0]{{\mathcal O}}

\newcommand{\opt}{\mathsf{Opt}}
\newcommand{\rank}{\mathsf{rank}}

\newcommand{\bfm}{\mathbf{m}}

\newcommand{\dist}{\operatorname{\rho}}

\newcommand{\WEEO}{\textsc{Euclidean Embedding Editing}\xspace}   

\newcommand{\EEO}{\textsc{Euclidean Embedding with Outliers}\xspace}   
\newcommand{\UEEO}{\textsc{UEEO}\xspace}

\newcommand{\EEDE}{\textsc{Euclidean Metric Violation Distance}\xspace}   

\DeclareMathOperator{\operatorClassNP}{NP}
\newcommand{\classNP}{\ensuremath{\operatorClassNP}\xspace}

\DeclareMathOperator{\operatorClassW}{W}
\newcommand{\classW}[1]{\ensuremath{\operatorClassW[#1]}}
\DeclareMathOperator{\operatorClassParaNP}{Para-NP\xspace}
\newcommand{\classParaNP}{\ensuremath{\operatorClassParaNP}\xspace}

\newlength{\RoundedBoxWidth}
\newsavebox{\GrayRoundedBox}
\newenvironment{GrayBox}[1]%
   {\setlength{\RoundedBoxWidth}{.93\textwidth}
    \def\boxheading{#1}
    \begin{lrbox}{\GrayRoundedBox}
       \begin{minipage}{\RoundedBoxWidth}}%
   {   \end{minipage}
    \end{lrbox}
    \begin{center}
    \begin{tikzpicture}%
       \node(Text)[draw=black!20,fill=white,rounded corners,%
             inner sep=2ex,text width=\RoundedBoxWidth]%
             {\usebox{\GrayRoundedBox}};
        \coordinate(x) at (current bounding box.north west);
        \node [draw=white,rectangle,inner sep=3pt,anchor=north west,fill=white] 
        at ($(x)+(6pt,.75em)$) {\boxheading};
    \end{tikzpicture}
    \end{center}}
    
\newenvironment{defproblemx}[2][]{\noindent\ignorespaces%
                                \FrameSep=6pt%
                                \parindent=0pt%
                \vspace*{-1.5em}
                \ifthenelse{\isempty{#1}}{%
                  \begin{GrayBox}{\textsc{#2}}%
                }{%
                  \begin{GrayBox}{\textsc{#2}  parameterized by~{#1}}%
                }
                \begin{tabular*}{\textwidth}{@{\hspace{.1em}} >{\itshape} p{1.8cm} p{0.8\textwidth} @{}}%
            }{
                \end{tabular*}%
                \end{GrayBox}%
                \ignorespacesafterend
            }  

\newcommand{\defproblem}[3]{
  \begin{defproblemx}{#1}
    Input:  & #2 \\
    Question: & #3
  \end{defproblemx}
}%
   
\usepackage[framemethod=tikz]{mdframed}

\definecolor{mycolor}{rgb}{0.122, 0.435, 0.698}
\newmdenv[innerlinewidth=0.02pt, roundcorner=4pt,linecolor=mycolor,innerleftmargin=6pt,
innerrightmargin=6pt,innertopmargin=6pt,innerbottommargin=6pt]{mybox}
\newmdenv[innerlinewidth=0.5pt, roundcorner=4pt,linecolor=black,innerleftmargin=6pt,
innerrightmargin=6pt,innertopmargin=6pt,innerbottommargin=6pt]{myboxblack}
\newmdenv[innerlinewidth=0.5pt, roundcorner=4pt,linecolor=mycolor,innerleftmargin=6pt,
innerrightmargin=6pt,innertopmargin=6pt,innerbottommargin=6pt]{myboxthick}

\longversion{\title{When Distances Lie:  Euclidean Embeddings in the Presence of Outliers and Distance Violations}} 

\author{Matthias Bentert}{University of Bergen, Norway}{matthias.bentert@uib.no}{}{Supported by the European Research Council (ERC) under the European Union's Horizon 2020 research and innovation programme (grant agreement No. 819416)}

\author{Fedor V. Fomin}{University of Bergen, Norway}{fedor.fomin@uib.no}{https://orcid.org/0000-0003-1955-4612}{Supported by the Research Council of Norway under BWCA project (grant no.~314528)}

\author{Petr A. Golovach}{University of Bergen, Norway}{petr.golovach@uib.no}{https://orcid.org/0000-0002-2619-2990}{Supported by the Research Council of Norway under BWCA project (grant no.~314528)}

\author{M. S. Ramanujan}{University of Warwick, UK}{r.maadapuzhi-sridharan@warwick.ac.uk}
{https://orcid.org/0000-0002-2116-6048}{Supported by Engineering and Physical Sciences Research Council (EPSRC) grant EP/V044621/1.}

\author{Saket Saurabh}{Institute of Mathematical Sciences, Chennai, India and University of Bergen, Norway}{saket@imsc.res.in}{https://orcid.org/0000-0001-7847-6402}{Supported by the European Research Council (ERC) under the European Union's Horizon 2020 research and innovation programme (grant agreement No. 819416); and Swarnajayanti
Fellowship grant DST/SJF/MSA-01/2017-18.}

\authorrunning{M. Bentert, F. V. Fomin, P. A. Golovach, M. S. Ramanujan, S. Saurabh} 

\Copyright{Matthias Bentert, Fedor V. Fomin, Petr A. Golovach, M. S. Ramanujan, Saket Saurabh} 



\category{} 

\relatedversion{} 

\EventEditors{John Q. Open and Joan R. Access}
\EventNoEds{2}
\EventLongTitle{42nd Conference on Very Important Topics (CVIT 2016)}
\EventShortTitle{}
\EventAcronym{}
\EventYear{}
\EventDate{December 24--27, 2016}
\EventLocation{Little Whinging, United Kingdom}
\EventLogo{}
\SeriesVolume{}
\ArticleNo{}

\hideLIPIcs

\begin{document}

\maketitle

\begin{abstract}
Distance geometry explores the properties of distance spaces that can be exactly represented as the pairwise Euclidean distances between points in $\mathbb{R}^d$ ($d \geq 1$), or equivalently, distance spaces that can be isometrically embedded in  $\mathbb{R}^d$.
In this work, we investigate whether a distance space can be isometrically embedded in  $\mathbb{R}^d$ after applying a limited number of modifications. Specifically, we focus on two types of modifications: \textit{outlier deletion} (removing points) and \textit{distance modification} (adjusting distances between points). The central problem, \WEEO, asks whether an input distance space on $n$ points can be transformed, using at most $k$ modifications, into a space that is isometrically embeddable in $\mathbb{R}^d$.

We present several fixed-parameter tractable (FPT) and approximation algorithms for this problem. Our first result is an algorithm that solves  \WEEO{} in time $(dk)^{\mathcal{O}(d+k)} + n^{\mathcal{O}(1)}$. The core subroutine of this algorithm, which is of independent interest, is a polynomial-time method for compressing the input distance space into an equivalent instance of \WEEO{} with $\mathcal{O}((dk)^2)$ points.

For the special but important case of \WEEO{} where only outlier deletions are allowed, we improve the parameter dependence of the FPT algorithm and obtain a running time of $\min\{(d+3)^k, 2^{d+k}\} \cdot n^{\mathcal{O}(1)}$. Additionally, we provide an FPT-approximation algorithm for this problem, which outputs a set of at most $2 \cdot \opt$ outliers in time $2^d \cdot n^{\mathcal{O}(1)}$. This 2-approximation algorithm improves upon the previous $(3+\varepsilon)$-approximation algorithm by Sidiropoulos, Wang, and Wang [SODA '17]. Furthermore, we complement our algorithms with hardness results motivating our choice of parameterizations.

The problem of embedding noisy distance data into Euclidean space has diverse applications, ranging from sensor network localization and molecular conformation to data visualization. To the best of our knowledge, apart from the work of Sidiropoulos, Wang, and Wang [SODA '17], our paper provides the first rigorous algorithmic analysis of this significant problem.

\end{abstract}

\section{Introduction}
 
The Euclidean Distance Matrix (EDM) is a matrix containing the squared Euclidean distances between points in a set. A central problem in Distance Geometry is determining whether a given matrix is an EDM \cite{blumenthal1970theory,liberti2017euclidean,dattorro2008convex,deza1997geometry}. 
That is, the task is to identify whether a  given distance space can be isometrically embedded into $\ell_2$-spaces, or equivalently, the pairwise Euclidean distances among points in $\mathbb{R}^d$ ($d \geq 1$). This problem  has a rich history, originating with Cayley~\cite{Cayley1841}, whose observations were formalized by Menger~\cite{Menger1928}, leading to Cayley-Menger determinants. Schoenberg~\cite{Schoenberg1935} further advanced the field by characterizing $\ell_2$-embeddable distances using negative type inequalities. Blumenthal's monograph~\cite{blumenthal1970theory} remains a foundational text, documenting the theoretical underpinnings of this area.

This problem has a wide range of applications, including sensor network localization \cite{patwari2005locating, doherty2001convex}, molecular conformation \cite{havel1985evaluation}, data visualization \cite{borg2005modern}, statistics \cite{everitt1997analysis}, psychology \cite{torgerson1958theory, shepard1962analysis}, learning image manifolds \cite{weinberger2004unsupervised}, handwriting recognition \cite{jain2004exploratory}, studying musical rhythms \cite{demaine2009distance}, and signal processing \cite{dokmanic2015euclidean}.
The complexity of determining whether a matrix is an EDM is well understood, and it can be efficiently addressed using Singular Value Decomposition (SVD) \cite{borg2005modern, torgerson1952multidimensional}.

In many practical applications involving EDMs, errors due to noise, missing values, or approximation inaccuracies are common. To address these challenges, several heuristic approaches based on Multidimensional Scaling, Low-Rank Matrix Approximation, and Semidefinite Programming have been developed to reconstruct the matrix and create an embedding that minimizes the impact of these errors. We refer readers to \cite{borg2005modern, dokmanic2015euclidean} for a comprehensive overview of the extensive literature on this topic.

Despite the importance and numerous practical applications of the EDM recognition problem with noise and errors, little was known about its computational complexity until recently. A notable exception is the work of Sidiropoulos, Wang, and Wang \cite{SidiropoulosWW17}, which initiates the study of EDM in the presence of outliers. The work of Sidiropoulos et al. serves as the initial foundation for our studies.

In this paper, we address the algorithmic question of minimizing the number of edits required to transform a given distance matrix into EDM. Specifically, we aim to apply the minimum number of modifications to a given \emph{distance space} \( (X, \dist) \) with distance function \( \dist \), such that the resulting distance space can be isometrically embedded into a \( d \)-dimensional Euclidean space \( \mathbb{R}^d \). We consider two types of editing operations.

The first operation is \emph{element deletion}. In matrix terms, this operation corresponds to deleting the row and column associated with the element we choose to remove. Following Sidiropoulos, Wang, and Wang \cite{SidiropoulosWW17}, we refer to the elements removed from the distance space as \emph{outliers}.

The second operation is \emph{distance modification}. Let \( X^{(2)} \) denote the set of unordered pairs of distinct elements in \( X \). For a pair of elements \( i, j \in X \), the modification operation alters the distance \( \dist(x, y) \). In terms of the distance matrix, this operation changes the ${ij}$-th and the ${ji}$-th entries. The problem of minimizing the number of modification operations to embed the resulting distance space in general metric spaces was studied in \cite{FanRB18,cohen2022fitting}, and embedding into ultrametric spaces was investigated in \cite{cohen2022fitting,charikar2024improved}.

Our primary algorithmic question is as follows: given integers \( \koutl \) and \( \kmod \), can a distance space \( (X, \dist) \) be transformed into a distance space that is embeddable in \( \mathbb{R}^d \) by removing at most \( \koutl \) outliers and performing at most \( \kmod \) modifications? In fact, we address an even more general weighted version of this problem, where each operation (outlier deletion or distance modification) is assigned a specific weight. More precisely, we study the following problem.

\medskip
\defproblem{\WEEO}
{Distance space $\Xcal=(X,\dist)$, weight functions $\woutl\colon X\rightarrow \mathbb{Z}_{\geq 0}$ and $\wmod\colon X^{(2)}\rightarrow \mathbb{Z}_{\geq 0}$, integers $W,\koutl,\kmod\geq 0$, and $d\geq 1$. }
{Is there a subset of at most $\koutl$ outliers $O\subseteq X$ and a set $D\subseteq X^{(2)}$ of at most $\kmod$ distances with $\woutl(O)+\wmod(D)\leq W$ such that the distance subspace $(X\setminus O, \dist')$ where $\dist'$ is obtained by modifying the values $\dist(x,y)$ for $\{x,y\}\in D$  is isometrically embeddable into $\mathbb{R}^d$?}

Two special cases of 
\WEEO  are of particular importance.
 The variant of the problem with $\kmod=0$, that is, 
of placing all but $\koutl$ outlier points of a distance space $(X,\dist)$ into a Euclidean space $\mathbb{R}^d$ of a given dimension  $d$ such that the Euclidean distance between any pair of points $x,$ is equal to  $\dist(x,y)$,  is called  \EEO \cite{SidiropoulosWW17}. For the variant with $\koutl=0$, that is, without deletion of outliers, following \cite{FanRB18}, we use the name \EEDE.

\medskip\noindent
EXAMPLE: The distance space $\Xcal=(X,\dist)$ with $X=\{1, \dots, 9\}$ and $\dist$ defined by distance matrix $D$, where the $ij$-th entry of $D$ is $\dist^2(i,j)$
\[
D=\begin{bmatrix}
        0 & \textcolor{red}{7} & 1 & 2 & 4 & 5 &\textcolor{blue}{1}& 4 & 5 \\
        \textcolor{red}{7} & 0 & 2 & 1 & 1 & 2 &\textcolor{blue}{10}& 5 & 4 \\
        1 & 2 & 0 & 1 & \textcolor{red}{11} & 4 &\textcolor{blue}{4}& 1 & 2 \\
        2 & 1 & 1 & 0 & 2 & 1 &\textcolor{blue}{8}& 2 & 1 \\
        4 & 1 & \textcolor{red}{11} & 2 & 0 & 1 &\textcolor{blue}{12}& 8 & 5 \\
        5 & 2 & 4 & 1 & 1 & 0 &\textcolor{blue}{7}& 5 & 2 \\
        \textcolor{blue}{1} & \textcolor{blue}{10} & \textcolor{blue}{4} & \textcolor{blue}{8} & \textcolor{blue}{12} & \textcolor{blue}{7 }&\textcolor{blue}{0}& \textcolor{blue}{8} & \textcolor{blue}{9} \\
        4 & 5 & 1 & 2 & 8 & 5 &\textcolor{blue}{8}& 0 & 1 \\
        5 & 4 & 2 & 1 & 5 & 2 &\textcolor{blue}{9}& 1 & 0 \\
        \end{bmatrix},
        \]
and  $\koutl=1$, $\kmod=2$, $d=2$, all weights equal to one, and $W=3$, is a yes-instance of \WEEO. Indeed, the 
 distance space $\Xcal'=(X',\dist')$ defined by distance matrix $D'$ in Fig.~\ref{fig:edm_embedding} that can be obtained from  
 $\Xcal$
by modifying distances between $\textcolor{red}{\{1,2\}}$ and 
$\textcolor{red}{\{3,5\}}$ and deleting outlier $\textcolor{blue}{7}$, is  isometrically embedable in $\mathbb{R}^2$.
  

 
 \begin{figure}[h!]
    \centering
    \begin{minipage}{0.45\textwidth}
        \[D'=
        \begin{bmatrix}
        {0} & 1 & 1 & 2 & 4 & 5 & 4 & 5 \\
        1 & 0 & 2 & 1 & 1 & 2 & 5 & 4 \\
        1 & 2 & 0 & 1 & 5 & 4 & 1 & 2 \\
        2 & 1 & 1 & 0 & 2 & 1 & 2 & 1 \\
        4 & 1 & 5 & 2 & 0 & 1 & 8 & 5 \\
        5 & 2 & 4 & 1 & 1 & 0 & 5 & 2 \\
        4 & 5 & 1 & 2 & 8 & 5 & 0 & 1 \\
        5 & 4 & 2 & 1 & 5 & 2 & 1 & 0 \\
        \end{bmatrix}
        \]
    \end{minipage}
    \hfill
    \begin{minipage}{0.5\textwidth}
        \centering
        \includegraphics[width=\textwidth]{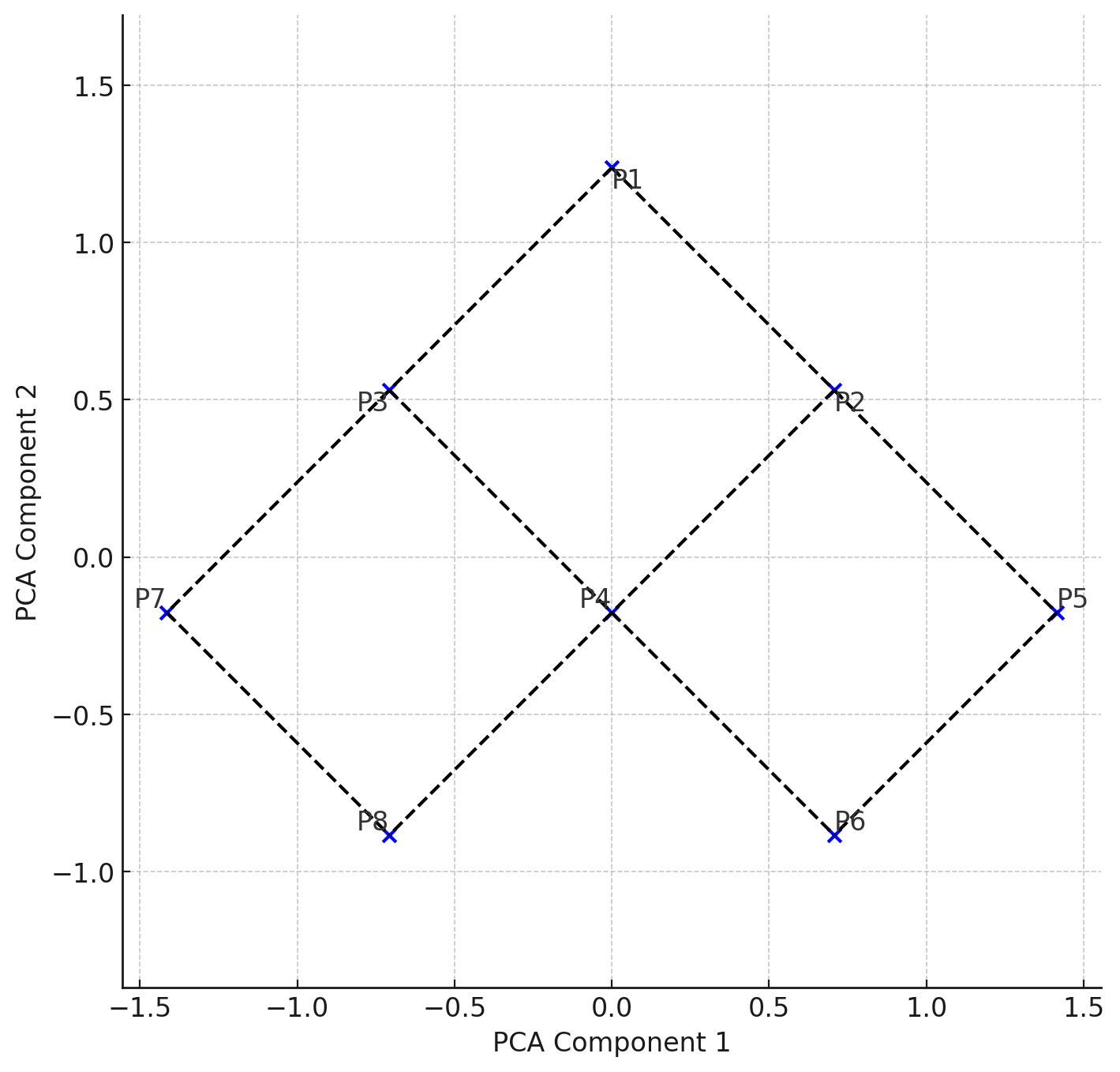}   
    \end{minipage}
     \caption{Matrix $D'$ and an embedding of its distance space   in $\mathbb{R}^2$.}
      \label{fig:edm_embedding}
\end{figure}

\subsection{Our results}
First, we present a compression  for \WEEO---a polynomial-time algorithm that reduces an instance of the problem to an equivalent instance with the number of points bounded by~$\Oh(k^2d^2)$ (\cref{thm:kernel}), where~$k = \koutl + \kmod$. As part of this compression, we also propose a~$(d+3)$-approximation algorithm for \EEO{}, which runs in polynomial time (\cref{lem:greedy}).
Using this compression algorithm, we design an FPT algorithm for \WEEO{} with a running time of~$(dk)^{\Oh(d+k)} + n^{\Oh(1)}$ (\cref{thmmodifFPT}). 

For \EEO{}, a particular but important special case of \WEEO{}---where the goal is to determine whether the distance space, excluding up to $k_O$ outliers, can be embedded in $\mathbb{R}^d$---we obtain better dependence on the parameters with a running time of~$\min\{(d+3)^{\koutl}, 2^{d+\koutl}\}\cdot n^{\Oh(1)}$ (\cref{thm:FPToutliers}). Furthermore, for this problem, we propose a 2-approximation algorithm (\cref{thm:2apprx}). This randomized algorithm, with a running time of~$2^d \cdot n^{\Oh(1)}$, guarantees a solution with at most $2 \cdot \opt$ outliers.
As is common in Computational Geometry, all our algorithms operate under the \emph{real RAM} computational model, assuming that basic operations over real numbers can be executed in unit time.

We also complement our algorithmic results with lower-bound proofs, establishing that both \EEO{} and \EEDE{} are \classNP-hard even when~$d=1$ (\cref{thm:outlier-d,thm:mod-d}). So, FPT algorithms for these problems parameterized by $d$ alone appears to be out of reach, motivating our FPT-approximation algorithm. Additionally, we prove that \EEO{} is W[1]-hard when parameterized by $\koutl$ alone (\cref{thm:w-hard}). These lower bounds indicate that to get an FPT algorithm for {\WEEO}, it is important that we parameterize by {\em both} $d$ and the solution size $\koutl+\kmod$. Importantly, our lower-bound results remain valid even in the unweighted case.

\subsection{Related work}
The computational complexity study of \EEO was initiated by Sidiropoulos, Wang, and Wang \cite{SidiropoulosWW17}. They demonstrated that, assuming the Unique Games Conjecture, the problem of computing a minimum outlier embedding into \(d\)-dimensional Euclidean space for any \(d \geq 2\) is NP-hard to approximate within a factor of \(2 - \varepsilon\) for any \(\varepsilon > 0\). On the algorithmic side, they showed that a \(2\)-approximation can be achieved in \(\Oh(n^{d+3})\) time. Additionally, they established that a \((3 + \varepsilon)\)-approximation is achievable in \((2/\varepsilon)^{d} d^{\Oh(1)} n^2 \log{n}\) time. For exact solutions, they noted an algorithm with a runtime of \(\Oh(n^{d+3}) + 2^{\koutl} n^2\).

They also obtained  results related to bicriteria approximation involving nonisomorphic embeddings, approximation of the number of outliers, and distortion.
The first result, is $(\Oh(\sqrt{\delta}),(2d+2)\koutl)$-relative outlier embedding in $\mathbb{R}^d$ in time $\frac{1}{\delta^{\Oh(1)}} 2^{\Oh(d)}n^2\log{n}$. The second result, $(\Oh(\sqrt{\delta}),2\koutl)$-relative  outlier embedding in time $\frac{1}{\delta^{\Oh(1)}} {\koutl}^{\Oh(d)}n^2$. 
(They define an algorithm as an  $\Oh(f (\delta), g(\koutl))$-relative outlier embedding of $\Xcal=(X,\dist)$ to $\mathbb{R}^d$, for some functions $f$ and $g$, if it either correctly decides that no embedding with distortion at most $\delta$ exists after removing $\koutl$ outliers, or outputs a set $Y$ of size $g(\koutl)$ such that there is an embedding $X\setminus Y$ into $\mathbb{R}^d$ 
of distortion $\Oh(f (\delta))$.) Since
allowing distortion could significantly decrease the number of outliers, 
our 2-approximation algorithm  for \EEO,  \Cref{thm:2apprx}, is incomparable with these results.

We are not aware of any algorithms with guaranteed performance for  \EEDE.  Parameterized and approximation algorithms for metric violation distance problems for 
general metric, tree distances  and ultrametric  could be found in 
 \cite{FanRB18,charikar2024improved,cohen2022fitting,gilbert2017sparse}.
More generally, embeddings of various metric spaces are a fundamental primitive in the design of algorithms \cite{Indyk01,indyk2004low,linial1995geometry,Linial02,arora2008euclidean,arora2009expander}, though prior work often focuses on minimizing embedding distortion.


\section{Preliminaries}\label{sec:prelim}

Let $X$ be a set. A function $\dist \colon X\times X \to \mathbb{R}_{\ge 0}$  is a \emph{distance} on $X$ if: 
(i) $\dist$ is symmetric, that is,  for any $x, y \in X$, $\dist(x, y) = \dist(y, x)$, and  
(ii) $\dist(x, x) = 0$ for all  $x\in X$.
Then, $ (X, \dist)$ is called a \emph{distance space}. 
%
If, in addition $\dist$ satisfies a triange inequality: $\dist(x, z) \le \dist(x, y) + \dist(y, z)$, for any $x, y, z \in X$, then $\dist$ is called a semimetric on $X$. And if $\dist(x,y)=0$ only for $x=y$, then $\dist$ is called a \emph{metric}, and $ (X,\dist)$ a \emph{metric space}.

More precisely, recall that 
for two points 
$p, q \in \mathbb{R}^d$, the Euclidean distance between $p$ and $q$ is  $\| p - q \|_2=\sqrt{\langle p,p\rangle+ \langle q,q\rangle - 2\langle p,q\rangle}$. 
We say that distance space $(X,\dist)$ is \emph{isometrically embeddable} into $\mathbb{R}^d$ if there is a map, called \emph{isometric embedding}, $\varphi \colon X\to \mathbb{R}^d$ such that 
$\dist(x,y)=\|\varphi(x)-\varphi(y)\|_2$
for all $x,y\in X$. Notice that we do not require $\varphi$ to be injective, that is, several points of $(X,\dist)$ may be mapped to the same point of $\mathbb{R}^d$.
Throughout the paper, whenever we mention an embedding, we mean an isometric embedding. Moreover, when we use the term $d$-embedding, we are referring to embedding into $\mathbb{R}^{d}$.  A $d$-embeddable  
  distance space is {\em strongly $d$-embeddable} if it is not $(d-1)$-embeddable. We use $X^{(2)}$ to denote the set of unordered pairs of two distinct elements of $X$. As convention, we assume that the empty set of points is $d$-embeddable 
for every $d$.

Let $\Xcal=(X,\dist)$ be a distance space where $X=\{x_1,\ldots,x_n\}$, and let $\dist_{i,j}=\dist(x_i,x_j)$ for all $i,j\in\{1,\ldots,n\}$. Then $\Xcal$ is equivalently defined by the \emph{distance matrix}  

\[
D(\rho)=\left(
 \begin{matrix}
0 & \dist_{1,2}^2 & \dist_{1,3}^2 & \dots & \dist_{1,n}^2 \\
\dist_{2,1}^2 & 0 & \dist_{2,3}^2 & \dots & \dist_{2,n}^2 \\
\dist_{3,1}^2 & \dist_{3,2}^2 & 0 & \dots & \dist_{3,n}^2 \\
\vdots & \vdots & \vdots & \ddots & \vdots \\
\dist_{n,1}^2 & \dist_{n,2}^2 & \dist_{n,3}^2 & \dots & 0 \\
\end{matrix} \right)
\]

We drop the explicit reference to $\rho$ when it is clear from the context and instead of $D(\rho)$, simply write $D$.
Suppose that $\Xcal$ is embeddable into $\mathbb{R}^d$. The ordered set 
$P=(p_1,\ldots,p_n)$ of points in $\mathbb{R}^d$ is said to be a \emph{realization} of $\Xcal$ if there is an embedding $\varphi \colon X\to \mathbb{R}^d$ such that $\varphi(x_i)=p_i$ for all $i\in\{1,\ldots,n\}$.
We use the well-known property (see e.g.~\cite{dokmanic2015euclidean}) that a realization is unique up to 
 \emph{rigid transformations} of $\mathbb{R}^d$, that is, distance-preserving transformations of Euclidean space (rotations, reflections, translations).

\begin{proposition}[\cite{dokmanic2015euclidean}]\label{prop:unique}
Let $(p_1,\ldots,p_n)$ and $(q_1,\ldots,q_n)$ be ordered set of points in $\mathbb{R}^d$. Then $\|p_i-p_j\|_2=\|q_i-q_j\|_2$ for all $i,j\in\{1,\ldots,n\}$ if and only if there is a rigid transformation of $\mathbb{R}^d$ mapping $p_i$ to $q_i$ for all $i\in\{1,\ldots,n\}$. 
\end{proposition}

A realization can be constructed (if it exists) in polynomial time.

\begin{proposition}[\cite{AlencarBLL15,SipplS85}]\label{prop:realization}
Given a distance space $\Xcal=(X,\dist)$ with $n$ points and a positive integer $d$, in $\Oh(n^3)$ time, it can be decided whether $\Xcal$ can be embedded into $\mathbb{R}^d$ and, if such an embedding exists, then a realization can be constructed in this running time.    
\end{proposition}
\newcommand{\mathbbr}[1]{\mathbb{R}^{#1}}

\begin{definition}[Metric basis]
    \label{def:metricBasis}
    Let $(X,\dist)$ be a $d$-embeddable distance space. A set $Y\subseteq X$ is a \emph{metric basis} if, given an isometric embedding $\varphi$ of $(Y,\dist)$ into $\mathbb{R}^d$, there is a unique way to extend $\varphi$ to an isometric embedding of $(X,\dist)$. Equivalently, if a 
realization of $(Y,\dist)$ is fixed then the embedding of any point of~$X\setminus Y$ in a $d$-embedding of $(X,\dist)$ 
is unique.  
\end{definition}

We will use the well-known characteristic of strong embeddability of a distance space into a Euclidean space. 
For $r+1$ points $x_0,x_1, \dots, x_r$ of distance space $(X, \dist)$
the \emph{Cayley-Menger determinant} is the determinant of the matrix obtained from the distance matrix by prepending a row and a column whose first element is zero and the other elements are one. Formally, let $\dist_{i,j}=\dist(x_i,x_j)$, $i,j\in \{0,\dots, r\}$. 
Then the Cayley–Menger determinant is  
\[
CM ( x_0,x_1, \dots, x_r) =  \det\left( \begin{matrix}
0 & 1 & 1 & 1 & \dots & 1 \\
1 & 0 & \dist_{0,1}^2 & \dist_{0,2}^2 & \dots & \dist_{0,r}^2 \\
1 & \dist_{0,1}^2 & 0 & \dist_{1,2}^2 & \dots & \dist_{1,r}^2 \\
1 & \dist_{0,2}^2 & \dist_{1,2}^2 & 0 & \dots & \dist_{2,r}^2 \\
\vdots & \vdots & \vdots & \vdots & \ddots & \vdots \\
1 & \dist_{0,r}^2 & \dist_{1,r}^2 & \dist_{2,r}^2 & \dots & 0 \\
\end{matrix} \right)
\]

\begin{proposition}[{\cite[Chapter~IV]{blumenthal1970theory}}]\label{thm:Blumental}
A distance space $\Xcal= (X, \rho)$ with $n$ points is strongly 
$d$-embeddable if and only if there exist $d+1$ points, say $X_d = \{ x_0, \ldots, x_d \}$, such that:
\begin{enumerate}
\item  $(-1)^{j+1}CM ( x_0,x_1, \dots, x_j)>0$ for $1\leq j\leq d$, and 
\item for any $x, y \in X\setminus X_d$,
\[CM ( x_0,x_1, \dots, x_d,x) =CM ( x_0,x_1, \dots, x_d,y)=CM ( x_0,x_1, \dots, x_d,x,y)=0.\] 
\end{enumerate}
Equivalently (see, for example,  \cite{SidiropoulosWW17}), 
$\Xcal$ is strongly 
$d$-embeddable if and only if there is a set of $d+1$ points $X_d=\{x_0,\ldots,x_d\}$ such that 
$(\{x_0,\ldots,x_j\},\dist)$ is strongly $j$-embeddable 
for all $j\in\{1,\ldots, d\}$, and
for every $x,y\in X\setminus X_d$, $(X_d\cup\{x\}\cup \{y\})$ is $d$-embeddable.
\end{proposition}

Thus, the embedding of the distance space $ \Xcal=(X, \rho)$ into $\mathbb{R}^d$ is essentially characterized by $d+1$ ``anchor'' points of $X$. Note that in a realization of $\Xcal$ that is strongly $d$-embeddable,  the anchor points correspond to a set of $d+1$ points in general position in $\mathbb{R}^d$.

Following~\cite{blumenthal1970theory}, we define sets of {\em independent points}.

\begin{definition}\label{def:independent}
Let $(X,\dist)$ be a 
distance space. For a nonnegative integer $r$, 
we say that $Y\subseteq X$ of size $r+1$ is \emph{independent} if $(Y,\dist)$ is strongly $r$-embeddable.      
\end{definition}
Then our  ``anchors'' are independent sets of size $d+1$. We use the following fact that the family of independent sets has matroid-like properties (implicit in~\cite[Chapter~IV]{blumenthal1970theory}).

\begin{proposition}[{\cite[Chapter~IV]{blumenthal1970theory}}]\label{prop:matroid}
Let $(X,\dist)$ be a strongly $d$-embeddable distance space. Then
\begin{itemize}
\item any single-element set is independent,
\item if $Y \subseteq X$ is independent, then any~$\emptyset \subset Z \subseteq Y$ is independent,
\item if $Y,Z \subseteq X$ are independent and $|Y|>|Z|$ then there is a~$y\in Y\setminus Z$ such that $Z\cup\{y\}$ is independent.
\end{itemize}
Furthermore, the maximum size of an independent set is $d+1$ and any independent set~$Y$ of size $d+1$ is a metric basis.
 \end{proposition}

In particular, \Cref{prop:matroid} immediately implies the following statement.

\begin{lemma}\label{equivalence}
     Let $\Xcal=(X,\dist)$ be a distance space. For every $Z\subseteq X$, if $Z$ is independent  
     and  for every $x,y\in X\setminus Z$, $(Z\cup\{x,y\})$ is $(|Z|-1)$-embeddable, 
     then $\Xcal$ is strongly $(|Z|-1)$-embeddable.
\end{lemma}

We also have the following  consequence of \Cref{thm:Blumental}.

\begin{lemma}\deferProof{}\label{lem:incrementDimension}
    Suppose $(X,\rho)$ is a $d$-embeddable distance space, $Y\subset X$ is independent and $x\in X\setminus Y$. Then, $(Y\cup\{x\},\rho)$ is $|Y|$-embeddable.
    \end{lemma}

\begin{proof}
    By definition, since $Y$ is independent, it follows that $(Y,\rho)$ is strongly $(|Y|-1)$-embeddable. Since $(X,\rho)$ is  $d$-embeddable, so is $(Y\cup\{x\},\rho)$. 
    Let $d'\leq d$ be such that $(Y\cup\{x\},\rho)$ is {\em strongly} $d'$-embeddable.   By  \Cref{thm:Blumental}, there exists a set $X_{d'}$ of $d'+1$ points in $Y\cup \{x\}$  such that  $(X_{d'},\rho)$ is strongly $d'$-embeddable. Since $X_{d'}\subseteq Y\cup \{x\}$, it follows that $d'+1\leq |Y|+1\implies d'\leq |Y|$. Hence, $(Y\cup\{x\},\rho)$ is $|Y|$-embeddable.
\end{proof}

To compress the weights in our compression algorithm in~\Cref{sec:kern}, we use the standard technique (see~\cite{EtscheidKMR17}) based on the algorithm of Frank and Tardos~\cite{FrankT87}.

\begin{proposition}[\cite{FrankT87}]\label{prop:FT}
There is an algorithm that, given a vector $w\in\mathbb{Q}^r$ and 
an integer~$N$, in (strongly) polynomial time finds a vector $\overline{w}\in\mathbb{Z}^r$ with $\|\overline{w}\|_{\infty}\leq 2^{4r^3}N^{r(r+2)}$ such that $\mathsf{sign}(w\cdot b)=\mathsf{sign}(\overline{w}\cdot b)$ for all vectors $b\in \mathbb{Z}^r$ with $\|b\|_1\leq N-1$.
\end{proposition}

\section{Compression for \WEEO}\label{sec:kern}
In this section, we show that, given an instance of \WEEO, one can in polynomial time construct an equivalent instance where the number of points is upper-bounded by a polynomial of $k = \koutl+\kmod$ and $d$ and, moreover, the encoding of the weights is also polynomial in these parameters.
\longversion{The main result of this section is the following theorem.}
\begin{theorem}\label{thm:kernel}
There is a polynomial-time algorithm that, given an instance \linebreak $(\Xcal=(X,\dist),\woutl,\wmod,W,\koutl,\kmod,d)$ of \WEEO, either solves the problem or 
constructs an equivalent instance $(\Xcal'=(X',\dist),\woutl',\wmod',W',\koutl',\kmod',d)$ such that $X'\subseteq X$, and for $k=\koutl+\kmod$, it holds that $|X'|=\Oh((kd)^2)$, $W'=2^{\Oh((kd)^{12})}$, $\woutl'(x)=2^{\Oh((kd)^{12})}$ for $x\in X'$, 
and $\wmod'(x,y)=2^{\Oh((kd)^{12})}$ for all $\{x,y\}\in X'^{(2)}$.
\end{theorem}

We remark that \Cref{thm:kernel} does not provide a kernelization algorithm according to the standard definition~\cite{cygan2015parameterized} as the size of the output instance is not upper-bounded by a function of the parameter---the distances between the points remain the same as in the input instance.   Moreover, since {\WEEO} is a decision problem, the word `solve' in the statement of \Cref{thm:kernel} technically only refers to deciding correctly whether the instance is a yes-instance. However, if the instance is determined to be a yes-instance, then our algorithm can also output a solution that witnesses this, i.e., the outlier set and modified new distance values (if any).  This additional feature of our compression algorithm will be used later in our FPT algorithm for {\WEEO}. 

The proof of \Cref{thm:kernel} is conducted in three steps: \begin{enumerate}\item  First, we design a polynomial-time $(d+3)$-approximation algorithm (\Cref{subsect=bootstrapping}). This algorithm allows us, given an instance of the problem,  either find a set of at most $(d+3)k$
points $A$ such that $(X\setminus A,\dist)$ is $d$-embeddable or conclude that we have a no-instance. \item In the first case, we iteratively construct a sequence of metric bases $Y_1,\ldots,Y_\ell$ for $(Z,\dist)$ using~\Cref{prop:matroid}, where initially $Z:=X\setminus A$ and in each iteration, we delete the points of the constructed basis from $Z$. If 
$X\setminus A$ is sufficiently large, then there is a subsequence of bases of the same size with at least $2(\koutl+2\kmod)+1$ elements. We select the first subsequence with this property. The crucial property exploited by our algorithm is that this subsequence uniquely defines embeddings of the majority of the points assuming that we have a yes-instance of the problem. We use this to give a series of reduction rules that identify points in $A$ that should be outliers in any solution and irrelevant points of $X\setminus A$ that could be deleted. 
\longversion{Thus, in polynomial time, we reduce the number of points to $\Oh((kd)^2)$.} 
\item Finally, by making use of \Cref{prop:FT}, we reduce the weights. 
\end{enumerate}

\subsection{Bootstrapping the Compression through an Approximation}\label{subsect=bootstrapping}

For a distance space $\Xcal=(X,\dist)$, a {\em $d$-outlier set} is a set $A$ of points such that $(X\setminus A,\dist)$ is $d$-embeddable. \longversion{Note that when we delete points from $X$, the function $\dist$ is implicitly restricted to $X\setminus A$.~}In unweighted {\EEO} (which we call {\UEEO}), the function $\woutl$ assigns unit weight to every point. 
Sidiropoulos et al.~\cite{SidiropoulosWW17} observed that this problem 
admits a simple greedy \((d+3)\)-approximation running in \(n^{\Oh(d)}\). 
\longversion{This observation is based on the following fact from distance geometry. 
Following Blumenthal \cite{blumenthal1970theory}, the \emph{order of congruence} for a \(d\)-dimensional space refers to the minimum number of points required such that their pairwise distances uniquely determine their geometric configuration (up to rigid transformations like translation, rotation, and reflection). 
By the result of Menger \cite{Menger1928}, for a \(d\)-dimensional Euclidean space \(\mathbb{R}^d\), the order of congruence is 
$n = d + 2$. 
This implies that any distance space admits an embedding into $\mathbbr{d}$ if and only if every subset of $d+3$ points admits such an embedding \cite{blumenthal1970theory}.  
Hence, one could enumerate every subset of $d+3$ points, identify a subset that cannot be embedded into  $\mathbbr{d}$ (if it exists), remove the points in this set and recurse.
}
We show that this algorithm can be sped up to run in \(n^{\Oh(1)}\) time, i.e., independent of $d$ in the exponent.
To do so, we give a polynomial-time subroutine that produces a small hitting set for all minimal solutions. This subroutine will later be used as the first step of our compression algorithm for {\WEEO}. Moreover, it will also be used to obtain an FPT algorithm for {\EEO} in Section \ref{sec:outlierFPT}.

\begin{lemma}\label{lem:blackBoxObstruction}
There is a polynomial-time algorithm that, given a distance space $\Xcal=(X,\dist)$ and an integer $d\geq 1$, outputs a set of points $\hat A$ of size at most $(d+3)$
such that if $\Xcal$ is not $d$-embeddable, then $\hat A$ intersects every inclusionwise minimal $d$-outlier set of $\Xcal$.
\end{lemma}

\begin{proof}
Assume that the distance space $\Xcal=(X,\dist)$ does not admit a $d$-embedding. Note that, in particular, this means that $|X|\geq 3$. We then do the following. 

\begin{mybox}
\begin{enumerate}\item Set $A:=\{p\}$ for an arbitrary point $p\in X$.

\item While 

$|A|\leq d$, do the following:
\begin{enumerate}
 \item If there is a point $x\in X\setminus A$ such that $(A\cup \{x\},\dist)$ is not $|A|$-embeddable,  then set $\hat A:=A\cup\{x\}$ and return $\hat A$.

\item If there is a point $x\in X\setminus A$ such that $A\cup\{x\}$ is independent, then set $A:=A\cup\{x\}$ and continue the loop, else exit the loop.
\end{enumerate}
\item If there are two distinct points $x,y\in X\setminus A$ such that $(A\cup \{x,y\},\dist)$ is not $(|A|-1)$-embeddable, then $\hat A:=A\cup\{x,y\}$ and return $\hat A$.
\end{enumerate}
\end{mybox}

Clearly, the algorithm runs in polynomial time. Since the while loop has at most $d$ iterations, the set $A$ when the loop is exited has size at most $d+1$ and so, the returned set~$\hat A$ has size at most $d+3$. 
\longversion{

Let us now argue that the algorithm is correct. 

Suppose that the set $\hat A$ is returned in Line 2(a). Since the loop invariant is that $A$ is independent, it follows that if $(\hat A,\dist)$ is $d$-embeddable, it must be $|A|$-embeddable (by \Cref{lem:incrementDimension}), which is not the case by construction. Since embeddability of a distance space into $\mathbb{R}^d$ is closed under deletion of points, we have that $(\hat A,\dist)$ is not $d$-embeddable
and therefore, $\hat A$ must intersect every minimal $d$-outlier set of $\Xcal$.

Now, assume that $\hat A$ is not returned in Line 2(a). 
When we exit the while loop, $A$ is independent. Moreover, either $|A|=d+1$ or for every point $x\in X\setminus A$, $(A\cup \{x\},\dist)$ is 
$|A|$-embeddable but not strongly $|A|$-embeddable  
and so, it is strongly $(|A|-1)$-embeddable. Furthermore, recall that $(X,\dist)$ is not $d$-embeddable. Then, it is also not $(|A|-1)$-embeddable since $A$ has size at most $d+1$.  
So, by \Cref{equivalence}, there exist points $x$ and $y$ that satisfy the check in Line 3 and so, the set $\hat A$ will be returned in this line.

Consider a minimal $d$-outlier set $O$ of $\Xcal$. It is sufficient to argue that if $\hat A$ is disjoint from $O$, then $(X\setminus O,\dist)$ is  $(|A|-1)$-embeddable. This would be a contradiction to our construction of $\hat A$ and \Cref{equivalence}. Suppose that $\hat A$ is disjoint from $O$ and $(X\setminus O,\dist)$ is not $(|A|-1)$-embeddable. Then, it must be the case that $(X\setminus O,\dist)$ is strongly $d'$-embeddable for some $|A|\leq d'\leq d$ and  we would have extended the set $A$ to a larger independent set in Line 2(b) (by  \Cref{prop:matroid}). 
}
\end{proof}

We can now use \Cref{lem:blackBoxObstruction} as a subroutine in the greedy approximation for the minimization variant of {\em unweighted} {\EEO}.

\begin{lemma}\deferProof{}\label{lem:greedy}
There is a polynomial-time algorithm that, given a distance space $\Xcal=(X,\dist)$ and an integer $d\geq 1$, outputs a set of points $A$ of size at most $(d+3)\cdot \opt$ such that $(X\setminus A,\dist)$ is embeddable into $\mathbb{R}^d$ where $\opt$ is the size of a smallest $d$-outlier set. 
\end{lemma}

\longversion{
\begin{proof}
Using \Cref{lem:blackBoxObstruction}, we iteratively construct disjoint sets $A_1,A_2,\ldots$ of at most $d+3$ points each such that each set should contain at least one outlier, that is, a point of any feasible solution to the instance and finally output the union of these sets.
Precisely, set $A_0=\emptyset$ and for every $i>0$, to compute $A_i$, we delete $A_1,\dots, A_{i-1}$ from $X$ and invoke \Cref{lem:blackBoxObstruction} on the remaining points. Because each set $A_i$ ($i>0$) contains at least one outlier and has size at most $d+3$, we obtain that the algorithm outputs a set of vertices $A$ of size at most $(d+3)\cdot \opt$ such that $(X\setminus A,\dist)$ is embeddable into $\mathbb{R}^d$. Because embeddability can be checked in polynomial time by~\Cref{prop:realization} and the algorithm of \Cref{lem:blackBoxObstruction} runs in polynomial time, we conclude that the overall running time of the algorithm is polynomial. This completes the proof. 
\end{proof}
}

\subsection{Compression}\label{subseccompression}

We are now ready to give the compression algorithm for {\WEEO}.

\begin{proof}[Proof of \Cref{thm:kernel}]
Let $(\Xcal=(X,\dist),\woutl,\wmod,W,\koutl,\kmod,d)$ be an instance of \WEEO. First, we preprocess the instance to reduce the number of points and then we explain how to reassign the weights.

Notice that if $(\Xcal,\woutl,\wmod,W,\koutl,\kmod,d)$ is a yes-instance, then it is possible to delete ${k=\koutl+\kmod}$ points to obtain a distance space embeddable in $\mathbb{R}^d$. 
In the first step of our algorithm, we call the algorithm from \Cref{lem:greedy} for the instance $(\Xcal,k,d)$ of \EEO{} {\em without} weights. 
This algorithm outputs a set of points $A$ of size at most $(d+3)\cdot \opt$ such that $(X\setminus A,\dist)$ is embeddable into $\mathbb{R}^d$ where $\opt$ is the minimum number of outliers. If $|A|>(d+3)k$ then we conclude that $(\Xcal,\woutl,\wmod,W,\koutl,\kmod,d)$ is a no-instance of \WEEO and stop. From now on, we assume that~$|A|\leq (d+3)(\koutl+\kmod)$. 

Let $Y=X\setminus A$.
We first show that if $Y$ has bounded size, then we can return the original instance of the problem and stop.

\begin{reduction}\label{red:small}
If $|Y|\leq 2(\koutl+2\kmod)(d+1)^2$ then return  $(\Xcal,\woutl,\wmod,W,\koutl,\kmod,d)$\longversion{ and stop}. \end{reduction}

We can apply this rule because if $|Y|\leq 2(\koutl+2\kmod)(d+1)^2$ then $|X|=|Y|+|A|\leq 
2(\koutl+2\kmod)(d+1)^2+(\koutl+\kmod)(d+3)\leq 5(\koutl+\kmod)(d+3)^2 = \Oh((kd)^2)$.
From now on, we assume that $|Y|>2(\koutl+2\kmod)(d+1)^2$.

Because $A$ is a feasible solution to \EEO, $(Y,\dist)$ can be embedded into $\mathbb{R}^d$. In the subsequent steps, we may delete some points in $A$ and reduce the parameter $\koutl$.  This may result in a trivial instance which we solve by applying the following rule whenever possible.

\begin{reduction}\label{rule:stop}
If $\koutl<0$ then return a no-answer and stop. If $\koutl\geq 0$ and $A=\emptyset$ then return a yes-answer and stop. 
\end{reduction}

We greedily partition $Y$ into sets $Y_1,\ldots,Y_\ell$ of size at most $d+1$ such that each~$Y_i$ is an inclusion-maximal independent set of points of $Y\setminus\bigcup_{j=1}^{i-1}Y_j$, that is, $Y_i$ is an inclusion-maximal subset such that $(Y_i,\dist)$ is strongly $(|Y_i|-1)$-embeddable. 

Assume that $i\geq 1$ and the sets $Y_1,\ldots,Y_{i-1}$ are already constructed. We set ${Z:=Y\setminus \bigcup_{j=1}^{i-1}Y_j}$. If $Z=\emptyset$, then the partition is constructed. 
If $Z\neq\emptyset$, then we do the following:
\begin{enumerate}
    \item To initiate the construction of $Y_i$, we set $Y_i=\{x\}$ for an arbitrary point $x\in Z$ and set $h=1$.
\item While there is a point $y\in Z\setminus Y_i$ such that 
$Y_i\cup\{y\}$ is independent, 
set $Y_i:=Y_i\cup\{y\}$ and $h:=h+1$.
\end{enumerate}
Because $Y$ is $d$-embeddable, $Z$ is strongly $t$-embeddable for some $t\leq d$ and contains an independent set of size $t$ by \Cref{thm:Blumental}. By~\Cref{prop:matroid}, the described procedure will construct a set $Y_i$ of size $t+1\leq d+1$ such that $(Y_i,\dist)$ is strongly $t$-embeddable.
Notice that by~\Cref{prop:matroid}, $Y_i$ is a metric basis for $(Z,\dist)$.

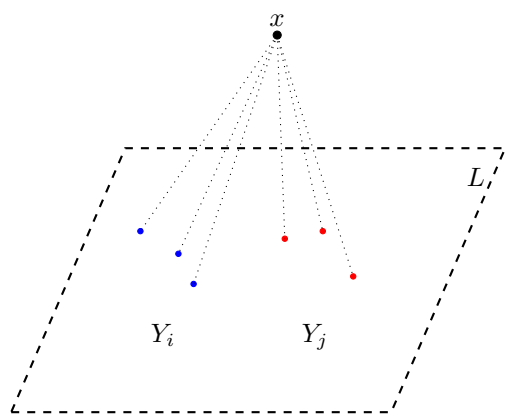
\begin{figure}[t!]
\centering
    \begin{tikzpicture}

    \coordinate (ClusterA) at (2, 2);
    \coordinate (ClusterB) at (4, 2);

    \node[below=20pt] at (ClusterA) {$Y_i$};
    \node[below=20pt] at (ClusterB) {$Y_j$};

    \coordinate (A1) at ($(ClusterA) + (0.2, 0.1)$);
    \coordinate (A2) at ($(ClusterA) + (-0.3, 0.4)$);
    \coordinate (A3) at ($(ClusterA) + (0.4, -0.3)$);
    \filldraw[blue] (A1) circle (1pt);
    \filldraw[blue] (A2) circle (1pt);
    \filldraw[blue] (A3) circle (1pt);

    \coordinate (B1) at ($(ClusterB) + (0.5, -0.2)$);
    \coordinate (B2) at ($(ClusterB) + (-0.4, 0.3)$);
    \coordinate (B3) at ($(ClusterB) + (0.1, 0.4)$);
    \filldraw[red] (B1) circle (1pt);
    \filldraw[red] (B2) circle (1pt);
    \filldraw[red] (B3) circle (1pt);

    \coordinate (A) at (0, 0);
    \coordinate (B) at (5.0, 0);
    \coordinate (C) at (6.5, 3.5);
    \coordinate (D) at (1.5, 3.5);

    \draw[thick, dashed, black] (A) -- (B) -- (C) -- (D) -- cycle;

    \node[below left=2mm of C] {$L$}; 

    \coordinate (X) at (3.5, 5);
    \node[above] at (X) {$x$};
    \filldraw[black] (X) circle (1.5pt);

    \foreach \p in {A1, A2, A3, B1, B2, B3} {
        \draw[dotted, thin] (X) -- (\p);
    }
\end{tikzpicture}
\caption{A pair $(Y_i,Y_j)$ of $x$-compatible sets. Notice that if 
$j\geq i$ and $|Y_j|\leq |Y_i|=t+1$ then $(Y_i\cup Y_j,\dist)$ admits a strong embedding into $\mathbb{R}^{t}$ by~\Cref{prop:matroid} because $Y_i$ is a metric basis for  $(Y_i\cup Y_j,\dist)$. Furthermore, given a realization of $(Y_i,\dist)$ such that the points are mapped into an affine subspace $L$ of dimension $t$, the points of $Y_j$ are mapped into $L$ and the mapping is unique. If $(Y_i\cup\{x\},\dist)$ is strongly $t$-embeddable then $x$ is embedded in $L$ and the embedding is unique. If $(Y_i\cup\{x\},\dist)$ is strongly $(t+1)$-embeddable, 
that is, $Y_i\cup\{x\}$ is independent then by~\Cref{prop:unique}, the distance space is embedded into a subspace $L'$ of dimension $t+1$ containing $L$ and the embedding of $x$ is unique up to the reflection with respect to $L$.}\label{fig:comp}
\end{figure}

Let $x\in A$. We say that $Y_i$ for $i\in\{1,\ldots,\ell\}$ is \emph{$x$-compatible} if $(Y_i\cup\{x\},\dist)$ is embeddable in $\mathbb{R}^d$. Otherwise, $Y_i$ is \emph{$x$-incompatible}. Similarly, for two distinct sets $Y_i,Y_j$ for $i,j\in\{1,\dots,\ell\}$, the pair $(Y_i,Y_j)$ is 
\emph{$x$-compatible} if $(Y_i\cup Y_j\cup\{x\},\dist)$ is embeddable in $\mathbb{R}^d$, and 
$(Y_i,Y_j)$ is \emph{$x$-incompatible} otherwise (see~\Cref{fig:comp}).

Notice that if $Y_i$ is $x$-incompatible then for each solution to the instance either $x$ is an outlier, or $Y_i$ contains an outlier, or at least one distance between the points $Y_i\cup\{x\}$ should be modified. 



For each $i\in\{1,\ldots,l\}$, $1\leq |Y_i|\leq d+1$. We partition the family~$\{Y_1,\ldots,Y_\ell\}$ of sets into parts $C_1,\ldots,C_{d+1}$ according to their size (some parts may be empty). Formally, for $h\in\{1,\ldots,d+1\}$, $C_h=\{Y_i\mid 1\leq i\leq \ell\text{ and }|Y_i|=h\}$.

Consider the part $C_h$ for some $h\in\{1,\ldots,d+1\}$. For a point $x\in A$, we say that two $x$-compatible sets $Y_i,Y_j\in C_h$ are \emph{$x$-equivalent} if the pair $(Y_i,Y_j)$ is $x$-compatible,  
that is, $(Y_i\cup Y_j\cup\{x\},\dist)$ is embeddable into $\mathbb{R}^d$.
We show the following claim.

\begin{claim}\deferProof{}\label{cl:equiv}
$x$-equivalency is an equivalence relation on the family of $x$-compatible sets in $C_h$. Furthermore, if $Y_i,Y_j\in C_h$ are $x$-equivalent and the pair $(Y_j,Y_s)$ for $Y_s\in C_{h'}$ with $h'\leq h$ is $x$-compatible then $(Y_i,Y_s)$ is $x$-compatible.      
\end{claim}

\longversion{
\begin{proof}[Proof of~\Cref{cl:equiv}]
Reflexivity and symmetry of the $x$-equivalency relation are trivial. To show transitivity, assume that $Y_i,Y_j\in C_h$ and $Y_j,Y_s\in C_h$ are $x$-equivalent. Recall that by the construction of the sets, $Y_i,Y_j,Y_s$ are inclusion-maximal independent sets of the same size $t+1$. Furthermore, by the construction of the sets $Y_i,Y_j,Y_s$, $(Y_i\cup Y_j\cup Y_s,\dist)$ is strongly $t$-embeddable. 
We fix a realization of $(Y_j,\dist)$ in $\mathbb{R}^t$. Then, 
by~\Cref{prop:matroid}, $Y_j$ is a metric basis 
and 
the realization of $(Y_i\cup Y_j\cup Y_s,\dist)$ is unique. 
Suppose that $(Y_j\cup\{x\},\dist)$ is strongly $t$-embeddable.  
Then the realization of $(Y_j\cup\{x\},\dist)$ is unique by \Cref{prop:matroid}. This means that  
$(Y_i\cup Y_s\cup\{x\},\dist)$ is $t$-embeddable, that is, $Y_i$ and~$Y_s$ are $x$-equivalent. Assume that $(Y_j\cup\{x\},\dist)$ is strongly $(t+1)$-
that is, $Y_j\cup\{x\}$ is independent. Then we extend the realization of $(Y_j,\dist)$ to the realization of $(Y_j\cup\{x\},\dist)$ and fix this realization. Again applying \Cref{prop:matroid}, we conclude that the realizations of $(Y_i\cup Y_j\cup\{x\},\dist)$ and 
$(Y_j\cup Y_s\cup\{x\})$ are unique and their union gives the realization of 
$(Y_i\cup Y_j\cup Y_s\cup \{x\},\dist)$. Thus, $Y_i$ and $Y_s$ are $x$-equivalent.

Notice that the same arguments prove the second part of claim because $(Y_j,\dist)$ is strongly $(|Y_j|-1)$-embeddable 
and, given a realization of  $(Y_j,\dist)$, the realization of  $(Y_i\cup Y_j\cup Y_s,\dist)$ in $\mathbb{R}^{|Y_j|-1}$ is unique because $h'\leq h$.
This completes the proof.
\end{proof}
}

The definition of $x$-equivalency implies the following property. 

\begin{claim}\label{cl:choice}
Let $R_1$ and $R_2$ be two distinct classes of $x$-equivalent $x$-compatible sets.
Then for any solution $(O, D)$ to the considered instance, 
either (i)~$x\in O$, or (ii)~for each set $Y_i$ in $R_1$ there is a~$y \in Y_i$ such that $y \in O$ or there exists a~$z$ such that~$\{y,z\} \in D$, or (iii)~for each set $Y_j$ in $R_2$ there is a~$y \in Y_j$ such that $y \in O$ or there exists a~$z$ such that~$\{y,z\} \in D$.
\end{claim}  

We use this property in the following rule. Because $|Y|>2(\koutl+2\kmod)(d+1)^{2}$ and each $Y_i$ is of size at most $d+1$,
there is $h\in\{1,\ldots,d+1\}$ with $C_h$ being of size at least $2(\koutl+2\kmod)+1$. Let $h^*$ be the maximum value of $h\in\{1,\ldots,d+1\}$ with this property.

\begin{reduction}\label{red:compatible}
If there is $x\in A$ such that each $x$-equivalence class $R$ of $x$-compatible sets of $C_{h^*}$ has size at most $|C_{h^*}|-(\koutl+2\kmod)-1$, then set $X:=X\setminus\{x\}$, $A:=A\setminus \{x\}$, and $\koutl:=\koutl-1$.
\end{reduction}

To show that the rule is safe, we claim that for any solution $(O,D)$ to the considered instance, $x\in O$ if the rule could be applied. Assume for the sake of contradiction that there is a solution $(O,D)$ such that $x\notin O$.  Because $|C_{h^*}|\geq 2(\koutl+2\kmod)+1$, there is~$Y_i\in C_{h^*}$ such that $Y_i\cap O=\emptyset$ and $\{y,z\}\cap Y_i=\emptyset$ for every $\{y,z\}\in D$. Then $Y_i$ is in some $x$-equivalency class $R$ of $x$-compatible sets. By \Cref{cl:choice}, for every set $Y_j$ from another class, either $Y_j$ contains an outlier or there is $\{y,z\}\in D$ such that $\{y,z\}\cap Y_j\neq\emptyset$.
Furthermore, for any $x$-incompatible set $Y_j\in C_{h^*}$, either $Y_j$ contains an outlier from $O$ or there is $\{y,z\}\in D$ with both $y,z\in Y_j\cup\{x\}$.
As $|R|\leq |C_{h^*}|-(\koutl+2\kmod)-1$, we obtain that for at least~$\koutl+2\kmod+1$ sets~$Y_j$ in~$C_h$,
either $Y_j$ contains an outlier or there is $\{y,z\}\in D$ with~$\{y,z\}\cap Y_j\neq\emptyset$.
However, this contradicts that $(O,D)$ is a solution as~$|O|+2|D|\leq \koutl+2\kmod$. The obtained contradiction proves safeness.

After the exhaustive application of \Cref{red:compatible}, we have that for each~$x\in A$, there is an $x$-equivalence class $R_x$ of $x$-compatible sets in~$C_{h^*}$ that contains at least $|C_{h^*}|-\koutl-2\kmod\geq (\koutl+2\kmod)+1$ sets
and all other classes in~$C_{h^*}$ contain at most $\koutl+2\kmod$ elements combined. We say that $R_x$ is \emph{large}.
We exhaustively apply the following rule using the fact that $R_x$ contains at least $(\koutl+2\kmod)+1$ sets.

\begin{reduction}\label{red:incompatible}
If there is $x\in A$ such that there are at least $\koutl+\kmod+1$ sets $Y_j\in C_h$ for any~$h\leq h^*$ such that $(Y_i,Y_j)$ is $x$-incompatible for some~$Y_i\in R_x$, then set $X:=X\setminus\{x\}$, $A:=A\setminus \{x\}$, and $\koutl:=\koutl-1$.
\end{reduction}

To argue safeness, we prove that $x\in O$ for any solution $(O,D)$ to the considered instance whenever the rule can be applied. For the sake of contradiction, assume that $x\notin O$ for some solution.
Note that by \Cref{cl:equiv}, if $(Y_i,Y_j)$ is $x$-incompatible for some $Y_i\in R_x$ then~$(Y_s,Y_j)$ is $x$-incompatible for all $Y_s\in R_x$. Because $R_x$ is large, there is $Y_s\in R_x$ such that~$O\cap Y_s=\emptyset$ and~$\{y,z\}\cap Y_s=\emptyset$ for all $\{y,z\}\in D$. Then, for every set $Y_j$ such that~$(Y_s,Y_j)$ is $x$-incompatible either $O\cap Y_j\neq \emptyset$ or there is $\{y,z\}\in D$ with both $y,z\in Y_j\cup\{x\}$. However, the number of such sets $Y_j$ is at least $(\koutl+\kmod)+1$ contradicting that 
$(O,D)$ is a solution. This proves that $x\in O$ and the rule is safe.

In the next crucial step of our algorithm, we identify \emph{important} sets  $Y_i$ for $\in\{1,\ldots,\ell\}$ and delete all other sets that are irrelevant.
\longversion{}
We apply the following \emph{marking} procedure that labels important sets $Y_i$ for $i\in\{1,\ldots,\ell\}$. For every $x\in A$, we do the following: 
\begin{itemize}
\item mark every $Y_i$ in $C_h$ for $h>h^*$,
\item for each $x\in A$, mark arbitrary $(\koutl+2\kmod)+1$ sets of $R_x$,
\item for each $x\in A$ and each $h\in\{1,\ldots,h^*\}$, mark each set $Y_j\in C_h$ such that there is a set~$Y_i\in R_x$ such that the pair $(Y_i,Y_j)$ is $x$-incompatible. 
\end{itemize}


We set $I=\{i\mid 1\leq i\leq \ell\text{ and }Y_i\text{ is marked}\}$ and $Y'=\bigcup_{i\in I}Y_i$, and define $X'=A\cup Y'$. 
\longversion{We prove the following crucial claim.}

\begin{claim}\deferProof{}\label{cl:irrelevant}
The instance $(\Xcal=(X,\dist),\woutl,\wmod,W,\koutl,\kmod,d)$ of \WEEO{} is equivalent to the instance~$(\Xcal'=(X',\dist),\woutl,\wmod,W,\koutl,\kmod,d)$.
\end{claim}

\longversion{
\begin{proof}[Proof of \Cref{cl:irrelevant}]
If $(O,D)$ is a solution to the instance $(\Xcal,\woutl,\wmod,W,\koutl,\kmod,d)$ then, trivially, $(O',D')$, where 
$O'=O\cap X'$ and $D'=D\cap X'^{(2)} $ is a solution to the instance \sloppypar$(\Xcal',\woutl,\wmod,W,\koutl,\kmod,d)$. Thus, we have to show the opposite implication.

Suppose that $(O,D)$ is a solution to $(\Xcal',\woutl,\wmod,W,\koutl,\kmod,d)$. We assume that the solution is inclusion-minimal. In particular, $D\subseteq (X'\setminus O)^{(2)}$. We claim that $(O,D)$ is a solution to $(\Xcal,\woutl,\wmod,W,\koutl,\kmod,d)$ as well. 

If $A\subseteq O$ then $(O,D)$ is a solution to $(\Xcal,\woutl,\wmod,W,\koutl,\kmod,d)$ because $(Y,\dist)$ is embeddable into $\mathbb{R}^d$. Assume that this is not the case and $\hat{A}=A\setminus O\neq \emptyset$. 
Let $\hat{Y}=Y'\setminus O$. We have that $(\hat{A}\cup\hat{Y},\dist')$ is $d$-embeddable where $\dist'$ is obtained from $\dist$ by modifying the values $\dist(y,z)$ for $\{x,y\}\in D$. 
We fix a realization $P$ of $(\hat{A}\cup\hat{Y},\dist')$ in $\mathbb{R}^d$ and denote by~$\varphi$ the corresponding isometric mapping of the distance space into $\mathbb{R}^d$.  
We prove that this realization can be extended to a realization $P'$ of $(X\setminus O,\dist')$ in $\mathbb{R}^d$.

The realization $P$ induces the realization $Q$ of $(\hat{Y},\dist')$. We show that $Q$ can be extended to a realization $Q'$ of $(Y\setminus O,\dist')$.

Observe that we marked at least $(\koutl+2\kmod)+1$ sets $Y_i\in C_{h^*}$. Then there is a marked set $Y_i\in C_{h^*}$ such that $O\cap Y_i=\emptyset$ and $\{y,z\}\cap Y_i=\emptyset$ for every 
$\{y,z\}\in D$.
Notice that $\varphi$ maps the points of $Y_i$ into $t$-dimensional affine subspace $L$ for $t=|Y_i|-1$ and recall that the embedding of $(Y_i,\dist)$ into $\mathbb{R}^t$ is strong. Denote by $S$ the realization of $(Y_i,\dist)$ induced by $Q$. 

Consider an arbitrary $y\in Y\setminus \hat{Y}$. Because $y\in C_h$ for some $h\leq h^*$, $(Y_i\cup\{y\},\dist)$ is strongly $t$-embeddable
by the construction of the sets $Y_1,\ldots,Y_\ell$. Then $y$ is embedded in $L$ and the embedding is unique by \Cref{prop:matroid}. This allows us to extend $S$ to the realization $S'$ of $(Y_i\cup(Y\setminus Y'),\dist)$ in a unique way. Then the union of $Q$ and $S'$ defines $Q'$. We extend $\varphi$ to the points of $Y\setminus Y'$ using $Q'$ and denote the obtained mapping $\psi$. 

To see that $\psi$ maps $(Y\setminus O,\dist')$ into $\mathbb{R}^d$ isometrically, consider two arbitrary points $y,z\in Y\setminus O$. If $y,z\in \hat{Y}$ then $\|\psi(y)-\psi(z)\|_2=\|\varphi(y)-\varphi(z)\|_2=\dist'(y,z)$.
If $y,z\in Y\setminus Y'$ then $\|\psi(y)-\psi(z)\|_2=\dist(y,z)=\dist'(y,z)$ by the construction of $Q'$ and the definition of the procedure constructing $Y_1,\ldots,Y_\ell$ as $(\bigcup_{h=1}^{h^*},\dist)$ is strongly $t$-embeddable. Suppose that $y\in Y\setminus Y'$ and $z\in \hat{Y}$. Recall that $\dist'(z,v)=\dist(z,v)$ for all $v\in Y_i$. If $\psi(z)=\varphi(z)\in L$, then the embedding of $z$ is unique. Since $(Y_i,\dist)$ is strongly $t$-embeddable in  and $(Y_i\cup\{y,z\},\dist)$ is $t$-embeddable, $\|\psi(y)-\psi(z)\|_2=\dist(y,z)=\dist'(y,z)$. Finally, assume that $\psi(z)=\varphi(z)\notin L$. Then, $(Y_i\cup \{z\},\dist)$ is strongly $(t+1)$-embeddable. Furthermore, $\varphi$ maps $z$ into $(t+1)$-dimensional subspace $L'$ containing $L$ and the mapping of $z$ is unique up to the reflection of $L'$ with respect to $L$ by \Cref{prop:unique}. Since the embedding of $y$ in $L$ is unique and 
$(Y_i\cup\{y,z\},\dist)$ is $(t+1)$-embeddable, $\|\psi(y)-\psi(z)\|_2=\dist(y,z)=\dist'(y,z)$.
This proves that $\psi$ is an isometric embedding of $(Y\setminus O,\dist')$ into $\mathbb{R}^d$.

Now we define the realization $P'$ of $(X\setminus O,\dist')$ in $\mathbb{R}^d$ as the union of $P$ and $Q'$. Recall that $\psi$ which is obtained by extending $\varphi$ to the points of $Y\setminus Y'$ maps the points of $X\setminus O$ to $\mathbb{R}^d$ according to $P'$. We show that   
$\psi$ maps $(X\setminus O,\dist')$ into $\mathbb{R}^d$ isometrically. 

We already proved that for any $y,z\in Y\setminus O$, $\|\psi(y)-\psi(z)\|_2=\dist'(y,z)$. Now we consider $x\in \hat{A}$ and $y\in Y\setminus O$. If $y\in \hat{Y}$ then $\|\psi(y)-\psi(z)\|_2=\|\varphi(x)-\varphi(y)\|_2=\dist'(y,z)$ by the definition of $\varphi$. Assume that $y\in Y\setminus Y'$.

By the definition of the marking procedure, there are $(\koutl+2\kmod)+1$ marked sets $Y_i\in R_x\subseteq C_{h^*}$. 
Then there is $Y_i\in R_x$ such that $O\cap (Y_i\cup\{x\})=\emptyset$ and $\{y,z\}\cap Y_i=\emptyset$ for all  
$\{y,z\}\in D$.  Because $y\in Y_j\in C_h$ for some $h\leq h^*$ where $Y_j$ is unmarked, the pair $(Y_i,Y_j)$ is $x$-compatible.   
Then can use the same arguments as above. We have that $\varphi$ maps the points of $Y_i$ into $t$-dimensional affine subspace $L$ for $t=|Y_i|-1$ and the embedding of $(Y_i,\dist)$ into $\mathbb{R}^t$ is strong. We fix the corresponding realization $S$ of $(Y_i,\dist)$.  
Then $(Y_i\cup\{y\},\dist)$ is strongly $t$-embeddable, $y$ is embedded in $L$, and the embedding is unique by \Cref{prop:matroid}. 
If $(Y_i\cup\{x,y\})$ is strongly $t$-embeddable, $\varphi$ maps $x$ into $L$ and this is a unique mapping by \Cref{prop:matroid}. Because 
$(Y_i,Y_j)$ is $x$-compatible, we have that  
$\|\psi(x)-\psi(y)\|_2=\dist(x,y)=\dist'(x,y)$. 
Suppose that $(Y_i\cup\{x\},\dist)$ is strongly $(t+1)$-embeddable. Then, $\varphi$ maps $x$ into $(t+1)$-dimensional subspace $L'$ containing $L$ and the mapping of $x$ is unique up to the reflection of $L'$ with respect to $L$ by \Cref{prop:unique}. Since the embedding of $y$ in $L$ is unique and $(Y_i,Y_j)$ is $x$-compatible, $\|\psi(x)-\psi(y)\|_2=\dist(x,y)=\dist'(x,y)$.
 This proves the claim.
\end{proof}
}

We have that the original instance $(\Xcal=(X,\dist),\woutl,\wmod,W,\koutl,\kmod,d)$ and the obtained instance $(\Xcal'=(X',\dist),\woutl,\wmod,W,\koutl,\kmod,d)$ of \WEEO are equivalent. We prove that $|X'|=\Oh(k^2d^2)$. 

\begin{claim}\deferProof{}\label{cl:size}
$|X'|\leq 9(\koutl+\kmod)^2(d+3)^2$.
\end{claim}

\longversion{
\begin{proof}[Proof of \Cref{cl:size}]
Observe that there are at most $d$ parts $C_h$ with $h>h^*$. By the choice of $h^*$, every such part contains at most $2(\koutl+2\kmod)$ sets. Thus, we have at most~$2(\koutl+2\kmod)d$ sets in total in  these parts. For each $x\in A$, we marked $\koutl+2\kmod+1$ sets in $R_x$. Then we marked at most 
$(\koutl+2\kmod+1)|A|\leq (\koutl+2\kmod+1)(\koutl+\kmod)(d+3)$ such sets. Finally, for each $x\in A$ and each $h\leq h^*$, we marked sets $Y_j$ such that $(X_i,Y_j)$ is $x$-incompatible.  Because \Cref{red:incompatible} is not applicable, we have that for each $x\in A$, we marked at most $\koutl+\kmod$ such sets, and $(\koutl+\kmod)^2(d+3)$ sets in total. Summarizing, we obtain that the total number of marked sets is at most
\begin{equation*}
2(\koutl+2\kmod)d+(\koutl+2\kmod+1)(\koutl+\kmod)(d+3)+(\koutl+\kmod)^2(d+3)\leq 8(\koutl+\kmod)^2(d+3).
\end{equation*}
Each set $Y_i$ contains at most $d+1$ points. Taking into account that $|A|\leq (\koutl+\kmod)(d+3)$,
we have that
\begin{equation*}
|X'|\leq 8(\koutl+\kmod)^2(d+3)(d+1)+(\koutl+\kmod)(d+3)\leq 9(\koutl+\kmod)^2(d+3)^2.    
\end{equation*}
This proves the claim.
\end{proof}
}

\begin{claim}\deferProof{}
 The overall running time of our algorithm is polynomial.
\end{claim}
\longversion{\begin{proof}
To evaluate the running time, recall that the algorithm from \Cref{lem:greedy} runs in  polynomial time. Further, the partition of $Y$ into the sets $Y_1,\ldots,Y_\ell$ can be done in polynomial time by making use of \Cref{prop:realization}. Then the $x$-equivalence classes for $C_{h^*}$ and all $x\in A$ can be constructed in polynomial time by \Cref{prop:realization}.  
This implies that \Cref{red:compatible}, \Cref{red:incompatible}, and the marking procedure can be done in polynomial time. Thus, the overall running time is polynomial. This concludes the first part of our algorithm where we reduce the number of points.
\end{proof}}

To reduce weights, we use \Cref{prop:FT}. We encode the weight functions and $W$ as an  
$m$-element vector where $m=n+\binom{n}{2}+1$ and $n$ is the number of points in $X'$. Then, we consider vectors $b\in\mathbb{Z}^m$ whose elements are from 
$\{-1,0,1\}$. Then $\|b\|_1\leq m$ and we set $N=m+1$. The algorithm from \Cref{prop:FT} gives weight functions $\woutl'\colon X'\rightarrow\mathbb{Z}_{\geq 0}$, $\wmod'\colon X'^{(2)}\rightarrow \mathbb{Z}_{\geq 0}$, and an integer $W'\geq 0$
 such that for every set of points $S\subseteq X$ and for every set of pairs of points $R\subseteq X'^{(2)}$, $\woutl(S)+\wmod(R)\leq W$ if and only if $\woutl'(S)+\wmod'(R)\leq W'$.  
Since~$n\leq 9(\koutl+\kmod)^2(d+3)^2$ by \Cref{cl:size}, we have that 
$W=2^{\Oh((kd)^{12})}$, $\woutl'(x)=2^{\Oh((kd)^{12})}$ for each~$x\in X'$, 
and $\wmod'(x,y)=2^{\Oh((kd)^{12})}$ for all $\{x,y\}\in X'^{(2)}$. 

This completes the proof. \end{proof} 

\section{FPT algorithm for \WEEO parameterized by solution size and dimension}\label{sec:fullFPT}

In this section, we give our FPT algorithm for \WEEO.

We prepare for it by recalling the relevant definitions and facts from \cite{BasuPR06}. In what follows, let $R$ be a real closed field and $\ell\in {\mathbb N}$. Let $\Pcal_\ell\subset R[X_1,\dots,X_t]$ be a finite set of $s$ polynomials each of degree at most $\ell$. 

\begin{definition}[\cite{BasuPR06}]\label{def:polynomialFormulas}
    A {\em $\Pcal_\ell$-atom} is one of $P=0$,$P\neq 0$, $P\ge 0$, $P\le 0$, where $P$ is a polynomial in $\Pcal_\ell$ and a {\em quantifier-free $\Pcal_\ell$-formula} is a formula constructed only from $\Pcal_\ell$-atoms together with the logical connectives $\wedge$, $\vee$ and $\neg$. 
\end{definition}

\begin{proposition}[Theorem 13.13, \cite{BasuPR06}]\label{prop:basuPolynomialModelFindingTheorem}
Let 
    $(\exists X_1)\dots  (\exists X_t) F(X_1,\dots,X_t),$
%
    be a sentence, where $F(X_1,\dots,X_t)$ is a quantifier free $\Pcal_{\ell}$-formula. There exists an algorithm to decide the truth of the sentence with complexity
    \footnote{The measure of complexity here is the number of arithmetic operations in the domain $D$.} 
    $s^{t+1}\cdot \ell^{\Oh(t)}$ in $D$ where $D$ is the ring generated by the coefficients of the polynomials in $\Pcal_{\ell}$. 
\end{proposition}

\longversion{We remind the reader that in the statement of \Cref{prop:basuPolynomialModelFindingTheorem}, $s$ is the number of polynomials in $\Pcal_\ell$. }


\begin{definition}
Consider\longversion{ $r+1$} points $x_0,x_1, \dots, x_r$ of distance space $(X, \dist)$ and a set $Z$ of pairs in $X$. For every $i<j\in \{0,\dots, r\}$, if $\{x_i,x_j\}\notin Z$, then $\hat \dist_{i,j}=\dist(x_i,x_j)$, otherwise $\hat \dist_{i,j}=z_{i,j}$ where $z_{i,j}$ is an indeterminate. 
Then the {\em $Z$-Augmented Cayley–Menger determinant} is  obtained from the Cayley-Menger determinant by replacing each $\dist_{i,j}$ with $\hat \dist_{i,j}$. That is, 
\longversion{
\[
CM_Z ( x_0,x_1, \dots, x_r) =  \det\left( \begin{matrix}
0 & 1 & 1 & 1 & \dots & 1 \\
1 & 0 & \hat \dist_{0,1}^2 & \hat \dist_{0,2}^2 & \dots & \hat\dist_{0,r}^2 \\
1 & \hat\dist_{0,1}^2 & 0 & \hat\dist_{1,2}^2 & \dots & \hat\dist_{1,r}^2 \\
1 & \hat\dist_{0,2}^2 & \hat\dist_{1,2}^2 & 0 & \dots & \hat\dist_{2,r}^2 \\
\vdots & \vdots & \vdots & \vdots & \ddots & \vdots \\
1 & \hat\dist_{0,r}^2 & \hat\dist_{1,r}^2 & \hat\dist_{2,r}^2 & \dots & 0 \\
\end{matrix} \right)
\]}

\end{definition}

\begin{observation}\label{obs:structureOfAugmentedCM}
Consider $r+1$ points $x_0,x_1, \dots, x_r$ of distance space $(X, \dist)$ and a set $Z$ of pairs in $X$. 
The $Z$-Augmented Cayley-Menger determinant is a multi-variate polynomial with real coefficients, over the set  $\{z_{i,j}\mid i<j,~\{x_i,x_j\}\in Z\}$ of indeterminates and where each monomial has degree at most $2\cdot (r+1)$. 
\end{observation}

 \begin{theorem}\label{thmmodifFPT}\label{thm:FPTWeeo}\WEEO  is solvable in~$(d(\koutl + \kmod))^{\Oh(d+ \kmod+\koutl )} + n^{\Oh(1)}$ time.
 \end{theorem}

\begin{proof}
Our FPT algorithm for {\WEEO} works as follows. We begin by running  the compression algorithm (\Cref{thm:kernel}). Following this, based on insights from  \Cref{thm:Blumental}, we reduce the task of solving the resulting instance to the task of testing a bounded number (in terms of $k,d$) of sentences with purely existential quantifications as described in \Cref{prop:basuPolynomialModelFindingTheorem}, where $s,\ell,t$ are all bounded by functions of $k$ and $d$. 

We now formalize this idea. 
Let $(\Xcal=(X,\dist),\woutl,\wmod,W,\koutl,\kmod,d)$ be the given instance of \WEEO. We run the algorithm of \Cref{thm:kernel} on this instance. Recall that this algorithm either solves the instance or produces an equivalent instance $(\Xcal'=(X',\dist),\woutl',\wmod',W',\koutl',\kmod',d)$ such that $X'\subseteq X$, and for $k=\koutl+\kmod$ and a bound $\tau(k,d)=2^{\Oh((kd)^{12})}$, we have that  $|X'|=\Oh((kd)^2)$. 
\longversion{Moreover, $W'\leq \tau(k,d)$, $\woutl'(x)\leq \tau(k,d)$ for $x\in X'$, and $\wmod'(x,y)\leq \tau(k,d)$ for all $\{x,y\}\in X'^{(2)}$.}
Note that since $X'$ has size bounded by $\Oh((kd)^2)$, we may assume that $\koutl'$ is bounded by $\bigoh((kd)^2)$ and $\kmod'$ is bounded by $\bigoh((kd)^4)$.

Since we now have a bounded number of points in $X'$, it is straightforward to guess a set $Z_O\subseteq X'$ of at most $\koutl'$ outliers and a subset $Z_M$ of pairs from $X''=X'\setminus Z_O$ of size at most $\kmod'$ such that (i) $\woutl'(Z_O)+\wmod'(Z_M)\leq W'$ and (ii) if there is a solution to the given instance, then there is a modification of the pairs in $Z_M$ such that $(X'',\rho'')$ is embeddable into $\mathbbr{d}$, where $\rho''$ differs from $\rho$ (restricted to $X''$) only among the pairs in $Z_M$. We next guess the embedding dimension of $(X'',\rho'')$, denoted by $r$. Note that $r\leq d$. 

 By \Cref{thm:Blumental}, we know that $(X'',\rho'')$ is $r$-embeddable if and only if there is a subset $\hat X=\{x_0,\dots,x_r\}$ of $X''$ such that  
 (i) $(-1)^{j+1}CM ( x_0,x_1, \dots, x_j)>0$ for $1\leq j\leq r$, and (ii)
 for any $x, y \in X''\setminus \hat X$,
$CM ( x_0,x_1, \dots, x_r,x) =CM ( x_0,x_1, \dots, x_r,y)=CM ( x_0,x_1, \dots, x_r,x,y)=0$, where the matrices are derived from the distance matrix $D(\rho'')$.  

We next guess $\hat X$. It remains to verify that there is a modification of the distances between each pair in $Z$ such that the resulting space $(X'',\rho'')$ is $r$-embeddable, i.e., satisfies the above two properties. To test this, we construct the formula $F$ obtained by taking the conjunction of the atoms below: 

\begin{enumerate}
	\item  $(-1)^{j+1}CM_{Z_{M}}( x_0,x_1, \dots, x_j)>0$, where $1\leq j\leq r$. 
	\item $CM_{Z_{M}}(x_0,x_1, \dots, x_r,x)=0$ for every $x\in X''\setminus \hat X$. 
	\item $CM_{Z_{M}}(x_0,x_1, \dots, x_r,x,y)=0$ for every $x,y\in X''\setminus\hat X$.
	\item $z_{i,j}\geq 0$ for each indeterminate defined by $Z_M$. 
\end{enumerate}

From \Cref{obs:structureOfAugmentedCM}, it follows that $F$ is a quantifier-free $\Pcal_{2(d+3)}$-formula, where $\Pcal_{2(d+3)}$ is a finite set of  polynomials, each of degree at most $2(d+3)$. 
From \Cref{thm:Blumental}, it follows that the sentence $\phi$  obtained by prepending $F$ with existential quantifications for every indeterminate $\{z_{i,j}\mid i<j,~\{x_i,x_j\}\in Z\}$
 is true if and only if there is a modification to the specified distances between pairs in $Z_M$ that results in $\rho''$ as required.
 Hence, we invoke the algorithm of \Cref{prop:basuPolynomialModelFindingTheorem} on $\phi$ and return ``yes'' if and only if this invocation returns ``yes''. 

We have already argued the correctness, so let us now analyze the running time of this algorithm.  Recall that we assume the Real-RAM model.

We have $\binom{(kd)^2}{\koutl}$ possibilities for $Z_O$ and at most $\binom{(kd)^4}{\kmod}$ choices for $Z_M$. We have at most $\binom{(kd)^2}{d}$ choices for $\hat X$, so ultimately, the number of tuples $(Z_O,Z_M,\hat X)$  we go over (and the number of invocations of \Cref{prop:basuPolynomialModelFindingTheorem})  is at most: \[\binom{(kd)^2}{\koutl}\cdot \binom{(kd)^4}{\kmod}\cdot \binom{(kd)^2}{d}=(kd)^{\bigoh(\koutl+\kmod+d)} =(kd)^{\bigoh(k+d)}.\]
The complexity of computing the formula $F$ defined by each choice of $Z_M$ and $\hat X$ is dominated by the time required to compute the determinants of $\bigoh((kd)^4)$ matrices, each of dimension at most $d+3$, and this is bounded by $(kd)^{\bigoh(d)}$. 
Each execution of the algorithm of \Cref{prop:basuPolynomialModelFindingTheorem} is on a sentence built from atoms contained in a set of most $\bigoh((kd)^4)$ polynomials, each of degree at most $2(d+3)$ and over the set $\{z_{i,j}\mid i<j,~\{x_i,x_j\}\in Z\}$ comprising at most $\kmod$ indeterminates (see \Cref{obs:structureOfAugmentedCM}). Plugging these bounds into \Cref{prop:basuPolynomialModelFindingTheorem}, we conclude that the number of field operations required for each execution is bounded by $(kd)^{\bigoh(\kmod)}\cdot d^{\bigoh(\kmod)}=(kd)^{\bigoh(\kmod)}$. 
  Since we take polynomial time to run the algorithm of \Cref{thm:kernel} at the beginning, the overall running time of our algorithm for {\WEEO} is as claimed: 
  $(d(\koutl + \kmod))^{\Oh(d+ \kmod+\koutl )} + n^{\Oh(1)}$
 \end{proof}

\section{Improving the parameter dependence for \EEO}\label{sec:outlierFPT}
In this section, we focus on \EEO{}, which is the special case of \WEEO{} where the goal is only to compute outliers.

\begin{theorem}\deferProof{}\label{thm:FPToutliers}
    \EEO can be solved in time $n^{O(1)}\cdot min\{(d+3)^k,2^{d+k}\}$. 
\end{theorem}

%

\begin{proof}

We present two algorithms with different performance guarantees depending on the relation between $k$ and $d$. In both cases, we give a {\em decision} algorithm for the {\em unweighted version}, but it will be straightforward to see that a solution can be computed in the same running time and that the algorithms extend to {\EEO} (i.e., with weights).
We describe the first algorithm ({\sf Alg-1}).
\vspace{5 pt}
 \begin{mybox}
\underline{{\sf Alg-1}$(X,\rho,k,d)$}:
\begin{enumerate}
    \item If $k<0$, then return ``no''.
    \item If $X$ is $d$-embeddable, 
    then return ``yes''. 
    \item Compute the set $\hat A$ given by \Cref{lem:blackBoxObstruction}. 
\item Return ``yes'' if and only if {\sf Alg-1}$(X\setminus \{x\},\rho,k-1,d)$ returns yes for some $x\in \hat A$.   
\end{enumerate}
 \end{mybox}
Since there are at most $(d+3)$ recursive calls made by each invocation to {\sf Alg-1} and the budget $k$ strictly decreases in each recursion, it follows that {\sf Alg-1}  runs in time $(d+3)^{k}\cdot n^{O(1)}$. The correctness follows from \Cref{lem:blackBoxObstruction}, the fact that the search is exhaustive and the fact that $d$-embeddability is closed under deletion of points.

We now describe the second algorithm ({\sf Alg-2}).
It takes as input a tuple $(X,\rho,k,d,Z)$, where $(X,\rho,k,d)$ is an instance of the problem and $Z\subseteq X$ is independent. 
The algorithm decides whether there is a set $A\subseteq X$ of size at most $k$ that is {\em disjoint} from $Z$ such that  $(X\setminus A,\rho)$ is $d$-embeddable. The initial invocation of the algorithm is with the tuple $(X,\rho,k,d,\emptyset)$. 

\vspace{5 pt}

 \begin{mybox}
\underline{{\sf Alg-2}$(X,\rho,k,d,Z)$}:
\begin{enumerate}
    \item If $k<0$ or $|Z|>d+1$ then return ``no''.
    \item If $X$ is $d$-embeddable, then return ``yes''. 
    \item If there is $z\not\in Z$ such that $Z\cup \{z\}$ is not $|Z|$-embeddable,  then return {\sf Alg-2}$(X\setminus\{z\},\rho,k-1,d,Z)$.
 \item  If there is $z\not\in Z$ such that $Z\cup \{z\}$ is independent then, if at least one out of  {\sf Alg-2}$(X\setminus\{z\},\rho,k-1,d,Z)$ and {\sf Alg-2}$(X,\rho,k,d,Z\cup \{z\})$ returns ``yes'', then return ``yes''; otherwise, return ``no''.

 
 \item If there are points $z_1,z_2\notin Z$ such that $Z\cup \{z_1,z_2\}$ is not $(|Z|-1)$-embeddable then, if at least one out of  {\sf Alg-2}$(X\setminus\{z_1\},\rho,k-1,d,Z)$ and {\sf Alg-2}$(X\setminus\{z_2\},\rho,k-1,d,Z)$ returns ``yes'', then return ``yes''; otherwise, return ``no''.

\end{enumerate}
 \end{mybox}

Note that in Line 4 of {\sf Alg-2}, we recurse into two exhaustive cases capturing whether or not $z$ is in the solution; in the former case, we delete it and in the later case, add it to $Z$. In Line 5 of the same algorithm, we recurse into two exhaustive cases capturing whether $z_1$ or $z_2$ is in the solution. 

Having completed the description of {\sf Alg-2}, we next analyze its running time and prove its correctness.
Since there are at most two recursive calls made by each invocation to {\sf Alg-2} and the measure $k+d+1-|Z|$ strictly decreases in each recursive call, the running time of {\sf Alg-2} is bounded by $2^{k+d}n^{O(1)}$. We next prove the correctness of {\sf Alg-2}.
We first argue that each step is correct and then argue that the case distinction is exhaustive. 
Lines 1 and 2 are trivially correct. Line 3 is correct as any outlier set disjoint from $Z$ must contain $z$.
This is due to the fact that otherwise, $(Z\cup \{z\},\rho)$ would be $d$-embeddable and by \Cref{lem:incrementDimension} then also $|Z|$-embeddable.  The two cases in Line 4 are by definition exhaustive. The branching in Line 5 is exhaustive due to \Cref{equivalence} and the fact that Line 4 is not executed.  
Finally, if none of the first four steps are executed (in particular, we have not solved the instance), then Line 5 must be executed due to \Cref{equivalence}. 

Since both algorithms are exhaustive search algorithms, it is straightforward to handle weights: when adding a point $x$ to the solution (Line 4 in {\sf Alg-2} and Lines 3-5 in {\sf Alg-2}), along with reducing $k$ by one, we can also reduce the weight budget $W$ by $\woutl(x)$ and add a check to both algorithms that ensures that $W$ is non-negative before returning ``yes''.
\end{proof}

\section{FPT-Approximation for Unweighted {\EEO}}\label{sec:approx}
In this section, we give an FPT-time 2-approximation for {\sc Unweighted} {\EEO} ({\UEEO}) parameterized by the dimension $d$. Recall that in the optmization version of this problem, the input is $(\Xcal=(X,\dist),d)$ where $(X,\dist)$ is a distance space and the goal is to output a $d$-outlier set of minimum size.


Let us first give a rough outline of our algorithm, motivating the central ideas.
Let $(\Xcal=(X,\dist),d)$ be an instance of {\UEEO}. Furthermore let $S_{\sf opt}$ be a minimum sized $d$-outlier set and $X^\star=X\setminus S_{\sf opt}$. Our algorithm is based on the following two ideas.

\subparagraph*{Idea 1:}  
Let us assume that $(X^\star,\dist)$ is strongly $d$-embeddable. Furthemore, let $U^\star\subseteq X^\star$ be a metric basis of $X^\star$ of size $d+1$. Let ${\cal C}_{\sf comp}$ be the set of elements $y \in X \setminus U^\star$ such that $(U^\star\cup \{y\},\dist)$ is  $d$-embeddable. Furthermore, let ${\cal C}_{\sf def}$ be the set of elements $y \in X \setminus U^\star$ such that $(U^\star\cup \{y\},\dist)$ is not $d$-embeddable. Observe that since $U^\star\subseteq X^\star$, 
we have that every element in ${\cal C}_{\sf def}$ belongs to $S_{\sf opt}$. On the other hand, each element in ${\cal C}_{\sf comp}$ can be $d$-embedded together with $U^\star$. However, it is possible that there could be two elements $x,y \in {\cal C}_{\sf comp} $ such that their positions in $\mathbb{R}^d$ are uniquely determined by $U^\star$, say    $\varphi(x)$ and $\varphi(y)$, respectively, but $\dist(x,y)\neq \|\varphi(x)-\varphi(y)\|_2$. We call such a pair $x,y$ an {\em incompatible pair}. 
This implies that since $U^\star\subseteq X^\star$, both $x$ and $y$ together cannot be in $X^\star$. This implies that either $x$ or $y$ belongs to $S_{\sf opt}$. This reasoning leads to an instance of the {\sc Vertex Cover} problem. This connection between {\EEO} and {\sc Vertex Cover} was also used in previous work \cite{SidiropoulosWW17}. Here, an input is a graph $G$ and the objective is to delete a minimum number of vertices, say $Q$ such that the graph has no edges left (that is $Q$ intersects every edge of $G$). We will be interested in solving the {\sc Vertex Cover} problem on a graph $G_{\sf comp}$, where $V(G_{\sf comp})= {\cal C}_{\sf comp} $, and there is an edge between two elements $x,y \in {\cal C}_{\sf comp} $  if they are an incompatible pair.
Note that $S_{\sf opt} \cap {\cal C}_{\sf comp} $ is a vertex cover of $G_{\sf comp}$. It is well known that {\sc Vertex Cover} admits a polynomial-time factor-$2$ approximation algorithm. For instance, one could take the endpoints of the edges in a maximal matching.  So,  using this algorithm we can obtain a set $Q'$ of size at most $2 |S_{\sf opt} \cap {\cal C}_{\sf comp}|$. This implies that ${\cal C}_{\sf def} \cup Q'$ is an outlier set of size at most $2 |S_{\sf opt}|$. 

What is not clear is how to obtain a metric basis of $X^\star$ that can play the role of $U^\star$ in the above discussion.

\subparagraph*{Idea 2:}
To address the difficulty of directly computing $U^\star$, we give a randomized algorithm that produces a sequence of independent sets $U_0,\dots,U_{d+1}$ along with a sequence  $A_1,\dots,A_{d+2}$ of $d$-outlier sets, where each $A_i$ is ``associated with'' $U_{i-1}$. We then show that with sufficiently high probability, either $U_{d+1}$ is the required $U^\star$ or one of the produced feasible solutions is the required factor-2 approximation. Roughly speaking, for each $U_{i-1}$, we compute $A_i$ as follows.
\begin{itemize}\item Any element that cannot be embbedded along with $U_{i-1}$ into $\mathbbr{d}$  is added to $A_i$. This is the set ${\cal C}_{\sf def}$ described in the previous paragraph.
\item Among the rest, we take the set ${\cal C}_{\sf comp}$ of compatible elements as described in the previous paragraph,  compute a 2-approximation for the associated {\sc vertex cover} instance and add it to $A_i$. 
\item Finally, we consider the remaining elements (besides ${\cal C}_{\sf def}$ and ${\cal C}_{\sf comp}$) call this set $W$, each of which when added to $U_{i-1}$ results in an independent set. At this point, we have two cases: either at least half of $W$ is in the solution (in which case we take all of this set in to the solution at a cost of factor 2) or less than half of $W$ is in the solution, in which case picking an element uniformly at random and adding it to $U_{i-1}$ gives us, with good probability, a strictly larger independent set $U_i$ in $X^\star$ and we move on to computing $U_{i+1}$.
    \end{itemize}

    Let us now formalize these ideas. 



\begin{theorem}\label{thm:2apprx}
    {\sc Unweighted} \EEO{} can be~$2$-approximated in~$2^d \cdot n^{\Oh(1)}$ time by a randomized algorithm with one-sided constant probability of error.
\end{theorem}

\begin{proof}

Let $S_{\sf opt}$ be a hypothetical minimum sized $d$-outlier set. Let $X^\star = X \setminus S_{\sf opt}$. 
Let us assume that $(X^\star,\dist)$ is strongly $d$-embeddable. Else, it is strongly $d'$-embeddable for some $d'\leq d$, in which case we guess $d'$ and set $d:=d'$. 
In what follows, we give a randomized algorithm that aims to obtain a metric basis of 
$(X^\star,\dist)$ if certain pre-conditions are satisfied. 
\longversion{We call this algorithm {\sf Alg-3}.}

\medskip

\begin{mybox}
\underline{{\sf Alg-3}$(X,\rho,d)$}:

\begin{enumerate}\item Set $U_{0}:=\emptyset$. 

\item For $i=1$ to $d+2$ do as follows. 
\begin{enumerate}
\item Construct ${\cal C}_{\sf comp}^i$, the set of elements $y \in X \setminus U_{i-1}$ such that $(U_{(i-1)} \cup \{y\},\dist)$ is  $(i-2)$-embeddable. 
 \item Construct ${\cal C}_{\sf def}^i$, the set of elements $y \in X \setminus U_{i-1}$ such that $(U_{(i-1)} \cup \{y\},\dist)$ is  not $d$-embeddable.
 \item Construct ${\cal C}_{\sf incomp}^i$, the set of elements $y \in X \setminus (U_{i-1}\cup {\cal C}_{\sf comp}^i\cup {\cal C}_{\sf def}^i)$ such that $(U_{(i-1)}\cup \{y\},\dist)$ is not $(i-2)$-embeddable. 
\item Select an element $u_i$ uniformly at random from ${\cal C}_{\sf incomp}^i$. 
\item Set $U_i:=U_{(i-1)}\cup \{u_i\}$. 

\end{enumerate}
\end{enumerate}
\end{mybox}







\longversion{We next establish some properties of {\sf Alg-3}. }

\begin{claim}\deferProof{}\label{clm:prepare}
For every $i\in \{0,\dots,d+1\}$, the following holds:
\begin{enumerate}
\item ${\cal C}_{\sf def}^i$ and ${\cal C}_{\sf comp}^i$ are disjoint.
\item Every $d$-outlier set of $(X,\rho)$ disjoint from $U_{i}$ contains ${\cal C}_{\sf def}^{i+1}$. 
\item  If for all $j\leq i$,  ${\cal C}_{\sf incomp}^j\neq \emptyset$, then $U_i$ is an independent set of size $i$. 
\end{enumerate}
\end{claim}
\longversion{
\begin{proof}[Proof of \Cref{clm:prepare}]
The first two statements are trivial. We prove the final statement using induction on $i$. The base case when $i=0$  holds trivially as $U_{0}=\emptyset$. Assume that $U_i$ is an independent set of size $i$ and let us argue this for $i+1$. By definition, for every $y\in {\cal C}_{\sf incomp}^{i+1}$, $(U_{i}\cup \{y\},\dist)$ is $d$-embeddable. Since  $U_{i}$ is an independent set of size $i$ by the induction hypothesis, Lemma~\ref{lem:incrementDimension} implies that $U_{i+1}$ is $i$-embeddable. However, by our construction we know that for every element $y \in {\cal C}_{\sf incomp}^{i+1}$ we have that $(U_{i}\cup \{y\},\dist)$ is not $(i-1)$-embeddable. In particular, this implies that $(U_{i}\cup \{u_{i+1}\},\dist)$ is not $(i-1)$-embeddable. Thus, we conclude that $U_{i+1}$ is strongly $i$-embeddable, that is, $U_{i+1}$ is an independent set. By construction, it has size one more than $|U_i|$, so it has size $i+1$ as claimed.  
\end{proof}
}

We next argue that for each $i\in [d+1]$, the independent set $U_i$ is preserved in $X^\star = X \setminus S_{\sf opt}$, with sufficiently high probability if certain conditions are met.

\begin{claim}\deferProof{}\label{claim:approxAlgo}
Let $i\in [d+1]$ such that for all $j\leq i$,  $|{\cal C}_{\sf incomp}^j \cap S_{\sf opt}| <  \frac{|{\cal C}_{\sf incomp}^j|}{2}$.
Then, $\Pr[U_i \subseteq X^\star] \geq  \frac{1}{2^i}$ and $U_i$ is an independent set of size $i$. 
\end{claim}

\longversion{
\begin{proof}[Proof of \Cref{claim:approxAlgo}]
The fact that $U_i$ is an independent set of size $i$ follows from \Cref{clm:prepare}~(3). Note that since we have a strict inequality in the premise of the statement, this implies that ${\cal C}_{\sf incomp}^j$ is non-empty for every $j\leq i$, thus satisfying the premise of  \Cref{clm:prepare}~(3).

It remains to prove that 
$\Pr[U_i \subseteq X^\star] \geq  \frac{1}{2^i}$. 
  We prove this using induction on $i$.  It is clear that when $i=1$, since at most half of ${\cal C}_{\sf incomp}^1$ is in the solution $S_{\sf opt}$, we have that  $ \Pr[U_1 \subseteq X^\star] \geq \frac{1}{2}$. Now, suppose that $i>1$.

  We have that,  $$ \Pr[U_i \subseteq X^\star] =  \Pr[U_{i-1} \subseteq X^\star \wedge u_i \in X^\star].$$  
   
   Recall that $u_i\in X\setminus U_{i-1}$.  Since the events $U_{i-1} \subseteq X^\star$ and $u_i \in X^\star$ are independent, we have the following. 
   $$
   \Pr[U_{i-1} \subseteq X^\star] \cdot \Pr[u_i \in X^\star]\geq \frac{1}{2^{i-1}} \cdot \frac{1}{2}= \frac{1}{2^i}.$$ 
  The bound on $\Pr[U_{i-1} \subseteq X^\star]$ follows from the induction hypothesis and the bound on  $\Pr[u_i \in X^\star]$ follows from the premise that $$|{\cal C}_{\sf incomp}^i \setminus S_{\sf opt}| \geq  \frac{|{\cal C}_{\sf incomp}^i|}{2}. $$
  This concludes the proof of the claim. 
\end{proof}
}


Next, we build a series of $d$-outlier sets  $A_1,\dots,A_{d+2}$ such that each $A_i$ is constructed based on $U_{i-1}$ and the sets constructed in the $i^{\rm  th}$ iteration of {\sf Alg-3}.

This is done in the following algorithm, call it {\sf Alg-4}\footnote{It is straightforward to ``embed'' {\sf Alg-4} into {\sf Alg-3} and give a single algorithm. However, we feel it is more insightful for the reader to see the two subroutines separately.}, which has access to the sets constructed by Alg-3.

 \begin{mybox}
\underline{{\sf Alg-4}$(X,\rho,d)$}:
\begin{enumerate}
\item For $i=1$ to $d+2$ do as follows. 
\begin{enumerate}
\item Create an instance of the {\sc Vertex Cover} problem on graph $G_{\sf comp}^i$, where $V(G_{\sf comp}^i)= {\cal C}_{\sf comp} ^i$, and there is an edge between two elements $x,y \in {\cal C}_{\sf comp} ^i$  if they are an incompatible pair. 
Let $O_i$ denote a minimum sized vertex cover of $G_{\sf comp}$. 
\item Using a known factor-$2$ approximation for {\sc Vertex Cover}, obtain a set 
 $Q_i$ such that $|Q_i|\leq 2 |O_i|$. 
\item $A_i:= Q_i \cup {\cal C}^{i}_{\sf incomp}\cup {\cal C}^{i}_{\sf def}$.
\end{enumerate}
\item Compute $ i^{\star} = \arg\min_{i \in [d+2]} |A_i|$. 
\item Return $A_{i^{\star}}$. 
\end{enumerate}
 \end{mybox}

\begin{claim}\deferProof{}\label{claim:approxAnalysis} Let $\ell \in [d+1]$ be the least integer such that  for all $j<\ell$, it holds that  $$|{\cal C}_{\sf incomp}^j \cap S_{\sf opt}| <  \frac{|{\cal C}_{\sf incomp}^j|}{2}. $$
Then, $A_\ell$ is a factor-2 approximation with probability at least $\frac{1}{2^{\ell}}$.
\end{claim}

\longversion{
\begin{proof}

We  have two cases for $\ell$ depending on whether or not the inequality in the premise of the claim 
holds also at $\ell$. \\ 

\noindent
{\bf Case 1:}   $$|{\cal C}_{\sf incomp}^\ell \cap S_{\sf opt}| <   \frac{|{\cal C}_{\sf incomp}^\ell|}{2}. $$

In this case, $\ell$ must be $d+1$.  By Claim~\ref{claim:approxAlgo}, the probability that $U_{d+1} \subseteq X^\star$ and  $U_{d+1}$ is an independent set of size $d+1$ is at least $\frac{1}{2^{d+1}}$.

By construction (see {\sf Alg-3}), ${\cal C}_{\sf comp}^{d+2}$ and ${\cal C}_{\sf def}^{d+2}$ partition the set $X\setminus U_{d+1}$ (since at $i=d+2$, $d$-embeddability is the same as $(i-2)$-embeddability). 
Moreover, from \Cref{clm:prepare}~(2), $S_{\sf opt}$ contains ${\cal C}^{d+2}_{\sf def}$, which is also contained in $A_{d+2}$ by definition. This implies that the size of $A_{d+2}\setminus {\cal C}^{d+2}_{\sf def}$ is $|Q_{d+2}|    \leq  2|O_{d+2}|$. Hence, we conclude that $|A_{d+2}| \leq  2 |S_{\sf opt}|$. 

\medskip
\noindent
{\bf Case 2:}  $$|{\cal C}_{\sf incomp}^\ell \cap S_{\sf opt}| \geq   \frac{|{\cal C}_{\sf incomp}^\ell|}{2}. $$

  By Claim~\ref{claim:approxAlgo}, the probability that $U_{\ell-1} \subseteq X^\star$ and  $U_{\ell-1}$ is an independent set of size $\ell-1$ is at least $\frac{1}{2^{\ell-1}}\geq \frac{1}{2^{d+1}}$. 



Consider the approximate solution $A_\ell$ we constructed. Observe that $S_{\sf opt} \cap {\cal C}_{\sf comp}^\ell $ is a vertex cover of $G_{\sf comp}^\ell$. Moreover, from \Cref{clm:prepare}~(2), $S_{\sf opt}$ contains ${\cal C}^{\ell}_{\sf def}$, which is also contained in $A_\ell$ by definition.

This implies that the size of $A_\ell\setminus {\cal C}^{\ell}_{\sf def}$ is:
\begin{eqnarray*}
|Q_\ell  \cup {\cal C}_{\sf incomp}|   & \leq & 2|O_\ell| + 2 |{\cal C}_{\sf incomp}^\ell \cap S_{\sf opt}|   \\
& \leq & 2 | S_{\sf opt} \cap {\cal C}_{\sf comp}^\ell| + 2 |{\cal C}_{\sf incomp}^\ell \cap S_{\sf opt}| \\
& = & 2 |S_{\sf opt}\setminus {\cal C}^{\ell}_{\sf def}| \text{ (Since ${\cal C}_{\sf comp}^\ell \cap {\cal C}_{\sf incomp}^\ell = \emptyset$)}
\end{eqnarray*}

Hence, we conclude that $|A_\ell| \leq  2 |S_{\sf opt}|$.  
\end{proof}
}

Since we compute \( i^{\star} = \arg\min_{i \in [d+2]} |A_i| \) and return $A_{i^{\star}}$,  by \Cref{claim:approxAnalysis} we are guaranteed that $|A_{i^{\star}}| \leq  2 |S_{\sf opt}| $ with probability at least $\frac{1}{2^{d+1}}$.


\smallskip
\noindent 
{\bf Running Time Analysis.} \longversion{ Clearly, the algorithm runs in polynomial time. } We have already argued that we succeed with probability $\frac{1}{2^{d+1}}$.  Thus, we can boost the success probability by independently running our polynomial-time algorithm $2^{d+1}$ times and returning the minimum size solution among these runs. Thus, the probability that the algorithm fails in all of the independent runs is upper bounded by 
\[\left(1-\frac{1}{2^{d+1}}\right)^{2^{d+1}} \leq 1- \frac{1}{e}.\]
Thus the algorithm succeeds with probability $1-(1/e) \geq 1/2$. 

Moreover, the total running time is upper bounded by $2^{d} n^{\Oh(1)}$. 
\end{proof}


\section{Lower bounds}\label{sec:lower}
In this section, we complement our algorithmic results by proving computational lower bounds. These lower bounds also motivate our choice of parameterizations. First, we show that \WEEO is \classParaNP-hard when parameterized by $d$. More precisely, we show that 
even the {\em unweighted} versions of both \EEO and \EEDE are \classNP-hard even for $d=1$.

\begin{theorem}\deferProof{}\label{thm:outlier-d}
\EEO is strongly \classNP-hard even for instances with unit weights, integer distances, and $d=1$.
\end{theorem}

\longversion{
\begin{proof}
We follow~\cite{SidiropoulosWW17} and reduce from the \textsc{Vertex Cover} problem which is well-known to be \classNP-complete~\cite{GareyJ79}. We recall that in \textsc{Vertex Cover}, the task is, given a graph $G$ and a positive integer $k$, to decide whether there is a vertex cover of size at most $k$, that is, a set of vertices $S$ of size at most $k$ such that every edge of $G$ has at least one of its endpoints in $S$. 

Let $G$ be an $n$-vertex graph and let $k\leq n$ be a positive integer. Denote by $v_1,\ldots,v_n$ the vertices of $G$. We construct the distance space $(X,\dist)$ as follows.
\begin{itemize}
\item Construct $k+2$ points $p_1,\ldots,p_{k+2}$ and set $\dist(p_i,p_j)=|i-j|$ for all $i,j\in \{1,\ldots,k+2\}$.
\item Construct $n$ points $x_1,\ldots,x_n$ which corresponds to the vertices of $G$.
\item For each $i\in\{1,\ldots,k+2\}$ and $j\in\{1,\ldots,n\}$, set $\dist(p_i,x_j)=i+j$.
\item For all distinct $i,j\in \{1,\ldots,n\}$, set 
$\dist(x_i,x_j)=
\begin{cases}
|i-j|&\mbox{if }v_iv_j\notin E(G),\\
0&\mbox{if }v_iv_j\in E(G).
\end{cases}
$
\end{itemize}
Trivially, the construction is polynomial time and $\dist(u,v)\leq 2n$ for all points $u,v$. We claim that $G$ has a vertex cover of size at most $k$ if and only if $(X=\{p_1,\ldots,p_{k+2}\}\cup\{x_1,\ldots,x_n\},\dist)$ admits an embedding into $\mathbb{R}$ after removal of at most $k$ points.

Assume that $G$ has a vertex cover $S$ of size at most $k$. We define the set of outliers 
$O=\{x_i\mid 1\leq i\leq n\text{ and }v_i\in S\}$. Consider the following mapping $\varphi\colon X\setminus O\rightarrow \mathbb{R}$:
\begin{itemize}
\item $\varphi(p_i)=-i$ for all $i\in\{1,\ldots,k+2\}$,
\item $\varphi(x_i)=i$ for all $x_i\in \{x_1,\ldots,x_n\}\setminus O$.
\end{itemize}
Because $G-S$ is an edgeless graph, we immediately obtain that $\varphi$ is an isometric embedding.

For the opposite direction, assume that there is a set of at most $k$ outliers $O$ such that there is an isometric embedding $\varphi\colon X\setminus O\rightarrow \mathbb{R}$. Because $|O|\leq k$, there are distinct $s,t\in\{1,\ldots,k+2\}$ such that $p_s,p_t\notin O$. Without loss of generality we can assume that $\varphi(p_s)=-s$ and $\varphi(p_t)=-t$. Then for any $x_i\in \{x_1,\ldots,x_n\}\setminus O$, $\varphi(x_i)=i$. Then for any two distinct $x_i,x_j\in \{x_1,\ldots,x_n\}\setminus O$, $|\varphi(x_i)-\varphi(x_j)|=|i-j|$, that is, $v_i$ and $v_j$ are not adjacent in $G$. Thus, $S=\{v_i\mid 1\leq i\leq n\text{ and }x_i\in O\}$ is a vertex cover of $G$. As $|O|\leq k$, we conclude that $S$ is a vertex cover of size at most $k$. This concludes the proof.
\end{proof}
}

\begin{theorem}\deferProof{}\label{thm:mod-d}
\EEDE is strongly \classNP-hard for the instances with unit weights, integer distances, and $d=1$.
\end{theorem}

\longversion{
\begin{proof}
We reduce from the \textsc{Max Cut} problem which is known to be \classNP-complete~\cite{GareyJ79}. In this problem, we are given a graph $G$ and a positive integer $\ell$, and the task is to decide whether $G$ admins a cut of size at least $\ell$, that is,  a partition $(A,B)$ of $V(G)$ such that the set of edges $E(A,B)$ having one endpoint in $A$ and the other in $B$ is of size at least $\ell$.

Let $G$ be an $n$-vertex graph and let $\ell\leq\binom{n}{2}$ be a positive integer. Denote by $v_1,\ldots,v_n$ the vertices of $G$. We construct the distance space $(X,\dist)$ as follows.
\begin{itemize}
\item Set $k=\binom{n}{2}-\ell$.
\item Construct $k+1$ points $p_0,\ldots,p_{k}$ and set $\dist(p_i,p_j)=0$ for all $i,j\in \{0,\ldots,k\}$.
\item Construct $n$ points $x_1,\ldots,x_n$ which corresponds to the vertices of $G$.
\item For each $i\in\{0,\ldots,k\}$ and $j\in\{1,\ldots,n\}$, set $\dist(p_i,x_j)=1$.
\item For all distinct $i,j\in \{1,\ldots,n\}$, set 
$\dist(x_i,x_j)=
\begin{cases}
2&\mbox{if }v_iv_j\in E(G),\\
1&\mbox{if }v_iv_j\notin E(G).
\end{cases}
$
\end{itemize}
It is straightforward that the construction is polynomial. Note that all the distances are at most two. We claim that $G$ has a cut of size at least $\ell$ if and only if $(X=\{p_0,\ldots,p_k\}\cup\{x_1,\ldots,x_n\},\dist)$ admits an embedding into $\mathbb{R}$ after modifying the value of $\dist(s,t)$ for at most $k$ pairs of points $s,t\in X$. 

Suppose that $G$ has a cut $(A,B)$ of size at least $\ell$. 
We define the mapping $\varphi\colon X\rightarrow \mathbb{R}$ as follows:
\begin{itemize}
\item $\varphi(p_i)=0$ for all $i\in\{0,\ldots,k\}$,
\item $\varphi(x_i)=-1$ for all $i\in \{1,\ldots,n\}$ such that $v_i\in A$,
\item $\varphi(x_i)=1$ for all $i\in \{1,\ldots,n\}$ such that $v_i\in B$.
\end{itemize}
Notice that for all $i,j\in \{0,\ldots,k\}$, $\varphi(p_i)-\varphi(p_j)=0=\dist(p_i,p_j)$. Also, for all $i\in\{1,\ldots,k\}$ and $j\in\{1,\ldots,n\}$, 
$|\varphi(p_i)-\varphi(x_j)|=1=\dist(p_i,x_j)$.
For $i,j\in \{1,\ldots,n\}$, we have that 
$|\varphi(x_i)-\varphi(x_j)|=2=\dist(x_i,x_j)$ if and only if $v_iv_j\in E(G)$ and the endpoints of the edge are in distinct sets of the partition $(A,B)$. Thus, to make $\varphi$ isometric, we have to modify the value $\dist(x_i,x_j)$ for $i,j\in \{1,\ldots,n\}$ such that $v_i$ and $v_j$ are either in the same set of the partition or $v_i$ and $v_j$ are in distinct sets and nonadjacent. The total number of such pairs $\{x_i,x_j\}$ is $\binom{n}{2}-|E(A,B)|\leq \binom{n}{2}-\ell=k$. 

For the opposite direction, assume that there is a mapping $\varphi\colon X\rightarrow \mathbb{R}$ that could be made isometric by modifying the value of $\dist(s,t)$ for at most $k$ pairs of points $s,t\in X$.
Because we can modify at most $k$ distances, there is $h\in\{1,\ldots,k\}$ such that $\dist(p_h,x_i)$ is not modified for all $i\in\{1,\ldots,n\}$. We assume without loss of generality that $\varphi(p_h)=0$. Then for any $i\in\{1,\ldots,n\}$, $\varphi(x_i)\in\{-1,1\}$. We define 
$A=\{v_i\mid 1\leq i\leq n\text{ and }\varphi(x_i)=-1\}$ and 
$B=\{v_i\mid 1\leq i\leq n\text{ and }\varphi(x_i)=1\}$.
Then for all distinct $i,j\in\{1,\ldots,n\}$, $|\varphi(x_i)-\varphi(x_j)|=\dist(x_i,x_j)$ if and only if 
$v_iv_j\in E(G)$ and the endpoints of this edge are in distinct sets of the partition $(A,B)$ of $V(G)$. Because at most $k$ distances are modified, we obtain that $G$ has a cut of size at least $\ell=\binom{n}{2}-k$. This concludes the proof. 
\end{proof}
}

Finally, we prove that \WEEO is \classW{1}-hard when parameterized by $\koutl+\kmod$ only. Moreover, the hardness holds even for the unweighted variant of \EEO. 
\longversion{For this, we have to restate the result of Fomin et al.~\cite{DBLP:journals/siamdm/FominGLS18} about the \textsc{Rank $h$-Reduction} problem. For us, it is convenient to define the problem as follows: given a matrix $M$ over $\mathbb{R}$ and positive integers $h$ and $k$, decide whether it is possible to obtain a matrix $M'$ from $M$ by deleting at most $k$ columns such that $\rank(M)-\rank(M')\geq h$.

\begin{proposition}[{\cite[Proposition~8.1]{DBLP:journals/siamdm/FominGLS18}}]\label{prop:rank}
\textsc{Rank $h$-Reduction} is \classW{1}-hard when parameterized by $h+k$ even when restricted to the instances with totally unimodular matrices.  
\end{proposition}

Fomin et al.~\cite{DBLP:journals/siamdm/FominGLS18} proved \Cref{prop:rank} for matrices over $\mathsf{GF}(2)$ representing \emph{cographic matroids} (we refer to the book~\cite{oxley2006matroid} for the definition and basic properties). However, cographic matroids are known to be \emph{regular}, that is, representable over any field. Also, it is known that a matroid is regular if and only if it can be represented by a totally unimodular matrix over $\mathbb{R}$~\cite{oxley2006matroid}. This immediately implies our version of \Cref{prop:rank}. Note that each element of a totally unimodular matrix is in $\{-1,0,1\}$.  
}

\begin{theorem}\deferProof{}\label{thm:w-hard}
\EEO parameterized by $k$ is \classW{1}-hard for $n$-point instances with unit weights and integer distance matrices whose entries are $\Oh(n)$.
\end{theorem}

\longversion{
\begin{proof}
We reduce from \textsc{Rank $h$-Reduction}. By \Cref{prop:rank}, the problem is \classW{1}-hard when parameterized by $k$ when restricted to the instances with matrices whose elements are in $\{-1,0,1\}$.
Let $(M,h,k)$ be such an instance and let $r=\rank(M)$. We assume that $M$ is a $(r\times n)$-matrix that has no zero columns, and we write $M=(\bfm_1,\ldots,\bfm_n)$ where $\bfm_1,\ldots,\bfm_n$ are the columns.  

We define $d=r-h$ and construct the distance space $(X,\dist)$ as follows.
\begin{itemize}
\item Construct $k+1$ points $p_0,\ldots,p_{k}$ and set $\dist(p_i,p_j)=0$ for all $i,j\in \{0,\ldots,k\}$.
\item Construct $n$ points $x_1,\ldots,x_n$ which corresponds to the columns of $M$.
\item For each $i\in\{0,\ldots,k\}$ and $j\in\{1,\ldots,n\}$, set $\dist(p_i,x_j)=\|\bfm_j\|_2$.
\item For all distinct $i,j\in \{1,\ldots,n\}$, set 
$\dist(x_i,x_j)=\|\bfm_i-\bfm_j\|_2$.
\end{itemize}
Clearly, the construction can be done in polynomial time. Notice that for any $u,v\in X=\{p_0,\ldots,p_k\}\cup\{x_1,\ldots,x_n\}$, $(\dist(u,v))^2\leq \max\{\max_{i\in\{1,\ldots,n\}}\|\bfm_i\|_2^2,\max_{i,j\in\{1,\ldots,n\}}\|\bfm_i-\bfm_j\|_2^2\}\leq 4n$ because each element of $M$ is in $\{-1,0,1\}$.

We claim that it is possible to obtain a matrix $M'$ with $\rank(M')\leq r-h$ by deleting at most $k$ columns of $M$ if and only if $(X,\dist)$ admits an embedding into $\mathbb{R}^d$ after deleting at most $k$ outliers. 

Suppose that there is a set $A$ of columns of $M$ of size at most $k$ such that the matrix $M'$ obtained from $M$ by deleting the columns of $A$ is of rank at most $d$. We set $O=\{x_i\mid 1\leq i\leq n\text{ and  }\bfm_i\in A\}$. We define the mapping $\varphi\colon X\setminus O\rightarrow \mathbb{R}^r$ as follows:
\begin{itemize}
\item $\varphi(p_i)=\mathbf{0}$ for all $i\in\{0,\ldots,k\}$,
\item $\varphi(x_i)=\bfm_i$ for all $x_i\in \{x_1,\ldots,x_n\}\setminus O$,
\end{itemize} 
where $\mathbf{0}$ denotes the zero vector. By the definition of $\dist$, we have that $\varphi$ is isometric. Because $\rank(M')\leq d$, the vectors $\bfm_i$ corresponding to the columns of $M'$ are in a $d$-dimensional subspace $L$. Then the points of the set $\{x_1,\ldots,x_n\}\setminus O$ are embedded in $L$ by $\varphi$. Furthermore, because $\mathbf{0}$ is in $L$, the points $p_0,\ldots,p_k$ are also mapped to $L$.
Thus, $\varphi$ isometrically maps $(X\setminus O,\dist)$ into $L$. This means that $(X\setminus O,\dist)$ is embeddable into $\mathbb{R}^d$.  

For the opposite implication, assume that there is a set $O$ of at most $k$ outliers such that $(X\setminus O,\dist)$. Because $O$ is of size at most $k$, there is $h\in\{0,\ldots,k\}$ such that $p_h\notin O$. Then  $(\{p_h\}\cup(\{x_1,\ldots,x_n\}\setminus O),\dist)$ is embeddable in $\mathbb{R}^d$. As $r\geq d$, the distance space is also embeddable in $\mathbb{R}^r$.
Consider the following mapping $\varphi\colon \{p_h\}\cup(\{x_1,\ldots,x_n\}\setminus O)\rightarrow \mathbb{R}^r$:
\begin{itemize}
\item $\varphi(p_h)=\mathbf{0}$,
\item $\varphi(x_i)=\bfm_i$ for all $x_i\in \{x_1,\ldots,x_n\}\setminus O$.
\end{itemize} 
This embedding is isometric by the definition $\dist$. Then, by~\Cref{prop:unique}, $\varphi$ gives a unique realization of $(\{p_h\}\cup(\{x_1,\ldots,x_n\}\setminus O),\dist)$ in $\mathbb{R}^r$ up to rigid transformations of $\mathbb{R}^r$. We have that $p_h$ is mapped to $\mathbf{0}$ and the points of 
$\{x_1,\ldots,x_n\}\setminus O$ are mapped to $\{\bfm_i\colon 1\leq i\leq n\text{ and }x_i\notin O\}$. Then because $(\{p_h\}\cup(\{x_1,\ldots,x_n\}\setminus O),\dist)$ is embeddable in $\mathbb{R}^d$, we have that the points $\mathbf{0}$ and $\bfm_i$ for $i\in\{1,\ldots,n\}$ such that $x_i\notin O$ are in a $d$-dimensional subspace. This means that $\rank(M')\leq d$ where $M'$ is the matrix obtained from $M$ by deleting columns $\bfm_i$ for $i\in\{1,\ldots,n\}$ such that $x_i\in O$. Since $|O|\leq k$, we got a matrix $M'$ with $\rank(M)\leq \rank(M)-h$ by deleting at most $k$ columns. This completes the proof.
\end{proof}
}

\bibliography{book_pc.bib}

\begin{thebibliography}{10}

\bibitem{AlencarBLL15}
Jorge Alencar, Tib{\'{e}}rius~O. Bonates, Carlile Lavor, and Leo Liberti.
\newblock An algorithm for realizing euclidean distance matrices.
\newblock {\em Electron. Notes Discret. Math.}, 50:397--402, 2015.
\newblock URL: \url{https://doi.org/10.1016/j.endm.2015.07.066}, \href {https://doi.org/10.1016/J.ENDM.2015.07.066} {\path{doi:10.1016/J.ENDM.2015.07.066}}.

\bibitem{arora2008euclidean}
Sanjeev Arora, James Lee, and Assaf Naor.
\newblock Euclidean distortion and the sparsest cut.
\newblock {\em Journal of the American Mathematical Society}, 21(1):1--21, 2008.

\bibitem{arora2009expander}
Sanjeev Arora, Satish Rao, and Umesh Vazirani.
\newblock Expander flows, geometric embeddings and graph partitioning.
\newblock {\em J. ACM}, 56(2):5, 2009.

\bibitem{BasuPR06}
Saugata Basu, Richard Pollack, and Marie-Fran\c{c}oise Roy.
\newblock {\em Algorithms in Real Algebraic Geometry (Algorithms and Computation in Mathematics)}.
\newblock Springer-Verlag, Berlin, Heidelberg, 2006.

\bibitem{blumenthal1970theory}
Leonard~Mascot Blumenthal.
\newblock {\em Theory and applications of distance geometry}.
\newblock Clarendon Press, Oxford, 1953.

\bibitem{borg2005modern}
Ingwer Borg and Patrick~J.F. Groenen.
\newblock {\em Modern Multidimensional Scaling: Theory and Applications}.
\newblock Springer, New York, 2005.

\bibitem{Cayley1841}
Arthur Cayley.
\newblock On the theory of determinants.
\newblock {\em Philosophical Magazine}, 19:1--16, 1841.

\bibitem{charikar2024improved}
Moses Charikar and Ruiquan Gao.
\newblock Improved approximations for ultrametric violation distance.
\newblock In {\em Proceedings of the 2024 Annual ACM-SIAM Symposium on Discrete Algorithms (SODA)}, pages 1704--1737. SIAM, 2024.

\bibitem{cohen2022fitting}
Vincent Cohen-Addad, Chenglin Fan, Euiwoong Lee, and Arnaud De~Mesmay.
\newblock Fitting metrics and ultrametrics with minimum disagreements.
\newblock In {\em Proceedings of the 63rd Annual Symposium on Foundations of Computer Science (FOCS)}, pages 301--311. IEEE, 2022.

\bibitem{cygan2015parameterized}
Marek Cygan, Fedor~V. Fomin, Lukasz Kowalik, Daniel Lokshtanov, D{\'a}niel Marx, Marcin Pilipczuk, Micha{\l} Pilipczuk, and Saket Saurabh.
\newblock {\em Parameterized Algorithms}.
\newblock Springer, 2015.
\newblock URL: \url{http://dx.doi.org/10.1007/978-3-319-21275-3}.

\bibitem{dattorro2008convex}
J.~Dattorro.
\newblock {\em Convex Optimization \& Euclidean Distance Geometry}.
\newblock Meboo Publishing USA, 2008.

\bibitem{demaine2009distance}
E.~D. Demaine, F.~Gomez-Martin, H.~Meijer, D.~Rappaport, P.~Taslakian, G.~T. Toussaint, T.~Winograd, and D.~R. Wood.
\newblock The distance geometry of music.
\newblock {\em Computational Geometry}, 42(5):429--454, July 2009.

\bibitem{deza1997geometry}
Michel Deza, Monique Laurent, and Robert Weismantel.
\newblock {\em Geometry of cuts and metrics}, volume~2.
\newblock Springer, 1997.

\bibitem{doherty2001convex}
L.~Doherty, K.~Pister, and L.~El~Ghaoui.
\newblock Convex position estimation in wireless sensor networks.
\newblock In {\em Proceedings of the IEEE Conference on Computer Communications (INFOCOM)}, volume~3, pages 1655--1663. IEEE, 2001.

\bibitem{dokmanic2015euclidean}
Ivan Dokmanic, Reza Parhizkar, Juri Ranieri, and Martin Vetterli.
\newblock Euclidean distance matrices: essential theory, algorithms, and applications.
\newblock {\em IEEE Signal Processing Magazine}, 32(6):12--30, 2015.

\bibitem{EtscheidKMR17}
Michael Etscheid, Stefan Kratsch, Matthias Mnich, and Heiko R{\"{o}}glin.
\newblock Polynomial kernels for weighted problems.
\newblock {\em Journal of Computer System Sciences}, 84:1--10, 2017.

\bibitem{everitt1997analysis}
B.~Everitt and S.~Rabe-Hesketh.
\newblock {\em The Analysis of Proximity Data}.
\newblock Arnold, London, 1997.

\bibitem{FanRB18}
Chenglin Fan, Benjamin Raichel, and Gregory~Van Buskirk.
\newblock Metric violation distance: Hardness and approximation.
\newblock In {\em Proceedings of the Twenty-Ninth Annual {ACM-SIAM} Symposium on Discrete Algorithms (SODA)}, pages 196--209. {SIAM}, 2018.

\bibitem{DBLP:journals/siamdm/FominGLS18}
Fedor~V. Fomin, Petr~A. Golovach, Daniel Lokshtanov, and Saket Saurabh.
\newblock Covering vectors by spaces: Regular matroids.
\newblock {\em {SIAM} J. Discret. Math.}, 32(4):2512--2565, 2018.
\newblock \href {https://doi.org/10.1137/17M1151250} {\path{doi:10.1137/17M1151250}}.

\bibitem{FrankT87}
Andr{\'{a}}s Frank and {\'{E}}va Tardos.
\newblock An application of simultaneous diophantine approximation in combinatorial optimization.
\newblock {\em Combinatorica}, 7(1):49--65, 1987.

\bibitem{GareyJ79}
Michael~R. Garey and David~S. Johnson.
\newblock {\em Computers and Intractability, A Guide to the Theory of {NP}-Completeness}.
\newblock W.H. Freeman and Company, New York, 1979.

\bibitem{gilbert2017sparse}
Anna~C. Gilbert and Lalit Jain.
\newblock If it ain't broke, don't fix it: Sparse metric repair.
\newblock In {\em 2017 55th Annual Allerton Conference on Communication, Control, and Computing (Allerton)}, pages 612--619. IEEE, 2017.

\bibitem{havel1985evaluation}
T.~F. Havel and K.~W{\"u}thrich.
\newblock An evaluation of the combined use of nuclear magnetic resonance and distance geometry for the determination of protein conformations in solution.
\newblock {\em Journal of Molecular Biology}, 182(2):281--294, 1985.

\bibitem{Indyk01}
Piotr Indyk.
\newblock Algorithmic applications of low-distortion geometric embeddings.
\newblock In {\em Proceedings of the 42nd IEEE Symposium on Foundations of Computer Science (FOCS)}, pages 10--33. IEEE, 2001.

\bibitem{indyk2004low}
Piotr Indyk and Jiri Matousek.
\newblock Low-distortion embeddings of finite metric spaces.
\newblock In {\em in Handbook of Discrete and Computational Geometry}, pages 177--196. CRC Press, 2004.

\bibitem{jain2004exploratory}
Viren Jain and Lawrence~K. Saul.
\newblock Exploratory analysis and visualization of speech and music by locally linear embedding.
\newblock In {\em Proceedings of the IEEE International Conference on Acoustics, Speech, and Signal Processing (ICASSP)}, volume~3, pages 984--987. IEEE, 2004.

\bibitem{liberti2017euclidean}
Leo Liberti and Carlile Lavor.
\newblock {\em Euclidean distance geometry}, volume~3.
\newblock Springer, 2017.

\bibitem{Linial02}
Nathan Linial.
\newblock Finite metric-spaces---combinatorics, geometry and algorithms.
\newblock In {\em Proceedings of the International Congress of Mathematicians, Vol. III}, pages 573--586, Beijing, 2002. Higher Ed. Press.

\bibitem{linial1995geometry}
Nathan Linial, Eran London, and Yuri Rabinovich.
\newblock The geometry of graphs and some of its algorithmic applications.
\newblock {\em Combinatorica}, 15(2):215--245, 1995.

\bibitem{Menger1928}
Karl Menger.
\newblock Untersuchungen Ã¼ber allgemeine metrik.
\newblock {\em Mathematische Annalen}, 100:75--163, 1928.

\bibitem{oxley2006matroid}
James~G. Oxley.
\newblock {\em Matroid theory}, volume~21 of {\em Oxford Graduate Texts in Mathematics}.
\newblock Oxford university press, 2nd edition, 2010.

\bibitem{patwari2005locating}
N.~Patwari, J.~N. Ash, S.~Kyperountas, A.~O. Hero, R.~L. Moses, and N.~S. Correal.
\newblock Locating the nodes: Cooperative localization in wireless sensor networks.
\newblock {\em IEEE Signal Processing Magazine}, 22(4):54--69, July 2005.

\bibitem{Schoenberg1935}
Isaac~J. Schoenberg.
\newblock Remarks to maurice frechet's article ``sur la definition axiomatique d'une classe d'espace distances vectoriellement applicable sur l'espace de hilbert''.
\newblock {\em Annals of Mathematics}, 36:724--732, 1935.

\bibitem{shepard1962analysis}
R.~Shepard.
\newblock The analysis of proximities: multidimensional scaling with an unknown distance function, part i.
\newblock {\em Psychometrika}, 27:125--140, 1962.

\bibitem{SidiropoulosWW17}
Anastasios Sidiropoulos, Dingkang Wang, and Yusu Wang.
\newblock Metric embeddings with outliers.
\newblock In {\em Proceedings of the Twenty-Eighth Annual {ACM-SIAM} Symposium on Discrete Algorithms (SODA)}, pages 670--689. {SIAM}, 2017.

\bibitem{SipplS85}
Manfred~J. Sippl and Harold~A. Scheraga.
\newblock Solution of the embedding problem and decomposition of symmetric matrices.
\newblock {\em Proc. Nat. Acad. Sci. U.S.A.}, 82(8):2197--2201, 1985.
\newblock \href {https://doi.org/10.1073/pnas.82.8.2197} {\path{doi:10.1073/pnas.82.8.2197}}.

\bibitem{torgerson1958theory}
W.~Torgerson.
\newblock {\em Theory and Methods of Scaling}.
\newblock Wiley, New York, 1958.

\bibitem{torgerson1952multidimensional}
Warren~S. Torgerson.
\newblock Multidimensional scaling: I. theory and method.
\newblock {\em Psychometrika}, 17(4):401--419, 1952.

\bibitem{weinberger2004unsupervised}
Kilian~Q. Weinberger and Lawrence~K. Saul.
\newblock Unsupervised learning of image manifolds by semidefinite programming.
\newblock In {\em Proceedings of the IEEE Conference on Computer Vision and Pattern Recognition (CVPR)}, volume~2, pages 988--995. IEEE, 2004.

\end{thebibliography}

\end{document}